\documentclass[11pt]{article}
\usepackage[utf8]{inputenc}
\usepackage[T1]{fontenc}

\usepackage[colorlinks=true]{hyperref}
\hypersetup{urlcolor=blue, citecolor=red, linkcolor=blue}

\usepackage{colortbl}

\usepackage{setspace}

\usepackage{graphicx}
\usepackage[active]{srcltx}

\usepackage{amsmath,amssymb}
\usepackage{bm}
\usepackage{enumerate}
\usepackage{amsthm}
\usepackage[all,2cell]{xy}
\usepackage{mathrsfs}
\usepackage{mathtools}
\mathtoolsset{showonlyrefs}

\usepackage{tikz-cd}
\usetikzlibrary{cd}

\usepackage{physics}
\usepackage{xcolor}
\DeclareUnicodeCharacter{0301}{*************************************}

\hypersetup{urlcolor=blue, citecolor=red, linkcolor=blue}

\usepackage[a4paper, left=20mm, right=20mm, top=25mm, bottom=25mm]{geometry}

\usepackage{marginnote}

\usepackage{imakeidx}
\makeindex[intoc]

\usepackage[nottoc]{tocbibind}
\usepackage[colorinlistoftodos]{todonotes} 

\theoremstyle{plain}
\newtheorem{theorem}{Theorem}[section]

\newtheorem{proposition}[theorem]{Proposition}

\newtheorem{lemma}[theorem]{Lemma}

\theoremstyle{definition}
\newtheorem{definition}[theorem]{Definition}

\newtheorem{example}[theorem]{Example}

\newcommand\restr[2]{{
  \left.\kern-\nulldelimiterspace 
  #1 
  \right|_{#2} 
}}

\newcommand{\R}{\mathbb{R}}
\renewcommand{\d}{\mathrm{d}}
\renewcommand{\dd}{\mathrm{d}}
\newcommand{\ee}{\text{e}}
\def \ii{{\rm i}}

\newcommand{\Cinfty}{\mathscr{C}^\infty}
\newcommand{\T}{\mathrm{T}}

\newcommand{\cT}{\mathrm{T}^\ast}
\newcommand{\Id}{\mathrm{Id}}

\DeclareMathOperator{\Ad}{Ad}

\newcommand{\Czero}{\mathscr{C}^0}
\newcommand{\Cone}{\mathscr{C}^1}

\newcommand*{\inn}[1]{\iota_{#1}}

\newcommand{\Lie}{\mathscr{L}}

\DeclareMathOperator{\Ima}{Im}

\DeclareMathAlphabet{\mathpzc}{OT1}{pzc}{m}{it}
\def\d{\mathrm{d}}

\DeclareMathOperator{\Ham}{Ham}
\DeclareMathOperator{\Hamloc}{Ham_{loc}}

\newcommand\xqed[1]{%
    \leavevmode\unskip\penalty9999 \hbox{}\nobreak\hfill
	\quad\hbox{#1}}
\newcommand\demo{\xqed{$\triangle$}}

\begin{document}

\parskip=3pt


\vspace{5em}

{\huge\sffamily\raggedright
\begin{spacing}{1.1}
    A symplectic approach to Schr\"odinger equations in the infinite-dimensional unbounded setting
\end{spacing}
}
\vspace{2em}

{\large\raggedright
    \today
}

\vspace{3em}

{\Large\raggedright\sffamily
    Javier de Lucas
}\vspace{1mm}\newline
{\raggedright
    Centre de Recherches Math\'ematiques, Universit\'e de Montr\'eal,\\ Pavillon André-Aisenstadt, 2920, chemin de la Tour,  Montréal (Québec) Canada  H3T 1J4.\\
    \medskip
    
    Department of Mathematical Methods in Physics, University of Warsaw, \\ ul. Pasteura 5, 02-093, Warszawa, Poland.\\
    e-mail: \href{mailto:javier.de.lucas@fuw.edu.pl}{javier.de.lucas@fuw.edu.pl} --- orcid: \href{https://orcid.org/0000-0001-8643-144X}{0000-0001-8643-144X}
}

\bigskip

{\Large\raggedright\sffamily
    Julia Lange
}\vspace{1mm}\newline
{\raggedright
    Department of Mathematical Methods in Physics, University of Warsaw, \\ul. Pasteura 5, 02-093, Warszawa, Poland.\\
    e-mail: \href{mailto:j.lange2@uw.edu.pl}{j.lange2@uw.edu.pl} --- orcid: \href{https://orcid.org/000-0001-6516-0839}{000-0001-6516-0839}
}

\bigskip

{\Large\raggedright\sffamily
    Xavier Rivas
}\vspace{1mm}\newline
{\raggedright
    Escuela Superior de Ingeniería y Tecnología, Universidad Internacional de La Rioja,\\Av. de la Paz, 137, 26006 Logroño, La Rioja, Spain.\\
    e-mail: \href{mailto:xavier.rivas@unir.net}{xavier.rivas@unir.net} --- orcid: \href{https://orcid.org/0000-0002-4175-5157}{0000-0002-4175-5157}
}

\vspace{3em}

{\large\bf\raggedright
    Abstract
}\vspace{1mm}\newline
{\raggedright
    By using the theory of analytic vectors and manifolds modelled on normed spaces, we provide a rigorous symplectic differential geometric approach to $t$-dependent Schr\"odinger equations on separable (possibly infinite-dimensional) Hilbert spaces determined by unbounded $t$-dependent self-adjoint Hamiltonians satisfying a technical condition. As an application, the Marsden--Weinstein reduction procedure is employed to map above-mentioned $t$-dependent Schr\"odinger equations onto their projective spaces. Other applications of physical and mathematical relevance are also analysed.
}
\bigskip

\medskip


\bigskip


\newpage

{\setcounter{tocdepth}{2}
\def\baselinestretch{1}
\small
\def\addvspace#1{\vskip 1pt}
\parskip 0pt plus 0.1mm
\tableofcontents
}

\pagestyle{myheadings}

\markright{{\rm
    J. de Lucas, J. Lange, and X. Rivas
}
- 
{\sl
    A symplectic approach to unbounded Schr\"odinger equations
}}

\section{Introduction}\label{ch:3}

The study of finite-dimensional symplectic manifolds has proven to be very fruitful in the description of mathematical and physical theories \cite{AM_78,Arn_89,Can_01,GS_96}. Among its applications, one finds the analysis of problems appearing in classical and quantum mechanics (see \cite{Arn_89,CCM_07,GS_96} and references therein). 

The analysis of symplectic manifolds modelled on (possibly infinite-dimensional) Banach spaces, the hereafter called {\it Banach symplectic manifolds}, is an interesting topic of research \cite{AMR_88,CL_09,CMP_90,CMP_90a,GKMS_18, Mar_68, Mar_72, OR_03,Tum_20}. Banach symplectic manifolds appear in relevant physical topics such as Schr\"odinger equations \cite{GKMS_18,Mar_68}, non-linear Schr\"odinger equations  \cite{CMP_90,CMP_90a,Fab_15}, and fluid mechanics \cite{Mar_68}.

The theory of Banach symplectic manifolds emerges as a generalisation of finite-dimensional symplectic geometry obtained by addressing appropriate topological and analytical issues due to the infinite-dimensional nature of general Banach spaces. Banach symplectic manifolds differ significantly from their finite-dimensional counterparts, e.g. there exists no immediate analogue of the Darboux theorem \cite{Mar_72}.  
Relevantly, quantum and classical problems related to infinite-dimensional Banach symplectic manifolds frequently require the study of non-smooth functions \cite{Kuk_95,Mar_68}.

Schr\"{o}dinger quantum mechanics represents a relevant field of application of symplectic geometry \cite{CCJL_19,CCM_07,CL_09, CL_11, CLR_09,CR_03,Mo_01}. Despite that, just a bunch of works study rigorously the differential geometric properties of $t$-dependent Schr\"{o}dinger equations on infinite-dimensional Hilbert manifolds (see \cite{AT_99,GKMS_18,Mar_68} and references therein). Most geometric works in the literature only deal with such Schr\"odinger equations on finite-dimensional Hilbert spaces \cite{CCJL_19,CCM_07,CL_09,CL_11} or on infinite-dimensional ones provided they are determined by a $t$-dependent bounded Hamiltonian operator. Eventually,  other additional simplifications are used \cite{AT_99,CDIM_17,CMP_90,CMP_90a,GKMS_18}. This significantly restricts the field of application of differential geometry in quantum mechanics. Instead, the development of an infinite-dimensional differential geometric approach allowing for the extension of previous ideas to $t$-dependent unbounded operators is extremely appealing. 

The lack of the searched extension is due, in part, to the lack of smoothness of the structures used to describe unbounded operators and other related structures, the problems to define differential geometric objects, like the commutator of operators, etcetera  \cite{Mar_68}. Known methods to deal with these problems rely on using pure functional analysis techniques with no differential geometric methods or the justification of technical assumptions and results on physical grounds. Works using differential geometry techniques in quantum mechanical models in infinite-dimensional manifolds do not generally provide a detailed explanation of the differential geometric constructions they use. Indeed, as far as we know, there is no rigorous differential geometric work dealing effectively with unbounded operators (cf. \cite{Gos_06}). In fact, some previous works treating unbounded operators may use function spaces, like the Schwarz space, that are not naturally normed \cite{Gos_11}, and are extremely difficult to describe with differential geometric techniques. It is worth noting that several research groups  may have their own conventions and assumptions. For instance, some researchers admit symplectic forms on infinite-dimensional manifolds to have a non-trivial kernel (see \cite[Remark 2.1]{Mo_01}), which goes against the classical definition, which we here follow.

This work applies infinite-dimensional geometric and functional analysis techniques to the rigorous geometrical study of $t$-dependent Schr\"odinger equations determined by (possibly unbounded) self-adjoint Hamiltonian operators on separable Hilbert spaces. Separable Hilbert spaces occur in most physical problems, and separability ensures the existence of a countable set of coordinates \cite{Hal_13}. As a first new contribution, we have applied the theory of normed spaces, which admits a relatively simple differential geometric analysis and extensions to the Banach case, and analytic vectors \cite{Car_60,FSSS_72,Goo_69,Nel_59} to improve the standard geometric approach given in \cite{Mar_68} to deal with unbounded Hamiltonian operators. In particular, we mostly restrict ourselves to $t$-dependent self-adjoint Hamiltonian operators taking values in a Lie algebra of operators with an appropriate common domain of analytic vectors. Although problems we are dealing with are not generally smooth, we will restrict ourselves to appropriate subspaces where problems are smooth enough to apply differential geometry. It is also worth noting that many of our practical applications restrict to the so-called strong Banach symplectic manifolds, where a Darboux theorem exists \cite{Wei_71} and other differential geometric techniques are applicable.

Our approach seems to be general enough to provide a rigorous theory while studying interesting physical systems, but it requires using some functional analysis techniques \cite{Mar_68}. We have relied as much as possible in differential geometric techniques and simplifying technical assumptions. In particular, self-adjoint operators and their mean values are described as objects on manifolds modelled on normed spaces steaming from the manifold structure of Hilbert spaces.  
Due to the fact that functional analysis is not well-known for most researchers working on differential geometric methods in physics, this work provides a brief introduction to the results on functional analysis and infinite-dimensional manifolds to be used. 

Our work solves some technical difficulties to ensure the standard differential geometric approach of the examples under study and other similar ones. For instance, we provide several technical details in the theory of analytic vectors related to the metaplectic representation, angular momentum operators, and other physical examples, which seem to be not developed elsewhere. Once these details are given, the further application of differential geometric techniques is simpler. In particular, our approach shows that the case of physical unbounded observables and other operators admitting a basis of eigenfunctions (without being necessarily symmetric), may be studied throughout our differential geometric techniques.

As an application of our techniques, the above-mentioned $t$-dependent Schr\"odinger equations are projected onto (possibly infinite-dimensional) projective spaces using a Marsden--Weinstein reduction. Our methods represent an appropriate framework, for instance, for the generalisation of the standard differential geometric works by G. Marmo and collaborators on quantum mechanics (see e.g. \cite{CCM_07,CL_11,CIMM_15,CMMGM_16,GKMS_18}) to the analysis of quantum mechanical problems on infinite-dimensional Hilbert spaces not determined by bounded operators. 

Let us detail the structure of this work. Section \ref{sec:2} is devoted to the basic functional analysis techniques and notations to be used hereafter. Differential geometry modelled on normed spaces \cite{Die_60,FB_66,Mar_68} and Banach spaces \cite{AMR_88,Lee_12} are discussed in Section \ref{sec:3}. In Section \ref{sec:4}, the theory of Banach and normed symplectic manifolds is briefly addressed. After studying weak and strong Banach symplectic manifolds, we survey their related Darboux theorems, the canonical one- and two-forms on cotangent bundles to Banach manifolds, and we finally introduce the Marsden--Weinstein reduction theorem. Section \ref{sec:5} provides a symplectic approach to infinite-dimensional separable Hilbert spaces. Section \ref{sec:6} focuses on vector fields induced by unbounded self-adjoint, observable, and diagonalisable operators, while Section \ref{sec:7} analyses Hamiltonian functions for quantum models. Section \ref{sec:8} studies a class of $t$-dependent Hamiltonian systems related to the $t$-dependent Schr\"odinger equations appearing in this work. 
In Section \ref{sec:9}, a Marsden--Weinstein reduction procedure is used to project the above-mentioned $t$-dependent Schr\"odinger equations onto an infinite-dimensional projective space. We also consider $t$-dependent Schr\"odinger equations related to unbounded $t$-dependent Hamiltonian operators as a special type of Hamiltonian systems on Hilbert manifolds.

\section{Fundamentals on operators}\label{sec:2}

Let us discuss several basic facts on functional analysis to be applied hereafter, e.g. in the mathematical formulation of Schr\"odinger quantum mechanics with unbounded operators (see \cite{AMR_88,Con_90,Die_60,FHHMZ_11,Hal_13} for details). 
 
We call  {\it operator} a linear map $A$ defined on a dense subspace $D(A)$ of a normed space $\mathcal{X}_1$ taking values on a normed space $\mathcal{X}_2$ of the form $A:D(A)\subset \mathcal{X}_1\rightarrow\mathcal{X}_2$. Operators are sometimes called {\it unbounded operators} to highlight that they do not need to be continuous.  It is worth stressing that defining an operator amounts to giving both its domain and its value on it. 
If two operators $A_i:D(A_i)\subset \mathcal{X}_1\rightarrow \mathcal{X}_2$, with $i=1,2$, satisfy that $D(A_1)\subset D(A_2)$ while $A_1$ and $A_2$ take the same values in $D(A_1)$, then we say that  $A_2$ is an {\it extension of $A_1$}. Two operators are equal, and we write $A_1=A_2$, when $A_1$ and $A_2$ have the same domain and they take the same values on it. 

To define the sum and composition of operators with possibly different domains, we proceed as follows. Let $A_i:D(A_i)\subset \mathcal{X}_1\rightarrow \mathcal{X}_2$, with $i=1,2$, be two operators. If  $D(A_1)\cap D(A_2)$ is dense in $\mathcal{X}_1$, then  $A_1+A_2$ is the operator on $D(A_1)\cap D(A_2)$ whose value at $\psi\in D(A_1)\cap D(A_2)$ is $A_1\psi+A_2\psi$. Similarly, given an operator $B:D(B)\subset \mathcal{X}_2\rightarrow \mathcal{X}_3$ such that $D_{A_1}=A_1^{-1}({\rm Im}(A_1)\cap D(B))$ is dense in $\mathcal{X}_1$, then we can define the composition $B\circ A_1:D_{A_1}\subset \mathcal{X}_1\rightarrow \mathcal{X}_3$. 

Let $H:D(H)\subset \mathcal{H}\rightarrow \mathcal{H}$ be an operator on a Hilbert space $\mathcal{H}$. Consider the associated operator $H_{\psi}: \phi\in D(H) \subset \mathcal{H} \mapsto \langle \psi \vert H \phi \rangle\in \mathbb{C}$. Whenever $H_{\psi}$ is  continuous on $D(H)$, the Riesz representation theorem \cite{Con_90} ensures that there exists $\psi_0 \in \mathcal{H}$ such that $H_{\psi}(\phi)= \langle \psi_0 \vert \phi \rangle$ for every $\phi\in D(H)$. If $H_{\psi}$ is continuous for every $\psi$ in a dense subset $D(H^\dagger)$ of $ \mathcal{H}$, then one can define the \textit{adjoint} of $H$, denoted by $H^{\dagger}$, as the operator  $H^{\dagger}:\psi\in D(H^\dagger)\subset \mathcal{H}\mapsto \psi_0\in \mathcal{H}$. In other words, $\langle \psi | H\phi\rangle = \langle H^\dagger\psi |\phi\rangle$ for every $\psi\in D(H^\dagger)$ and $\phi\in D(H)$.

Let us now detail certain relevant types of unbounded operators.
\begin{definition}
	An unbounded operator $H: D(H) \subset \mathcal{H}\rightarrow \mathcal{H}$ is called:
	\begin{enumerate}[a)]
		\item \textit{symmetric}: if $\langle \psi \vert H \phi \rangle = \langle H \psi \vert \phi \rangle$ for all $\phi, \psi \in D(H)$,
				\item \textit{skew-symmetric}: if $\langle \psi \vert H \phi \rangle = -\langle H \psi \vert \phi \rangle$ for all $\phi, \psi \in D(H)$,
		\item \textit{self-adjoint}: if $H^{\dagger}=H$,
				\item \textit{skew-self-adjoint}: if  $H^{\dagger}=-H$,
		\item \textit{essentially self-adjoint}: if $H$ admits a unique self-adjoint extension,
  \item \textit{observable}: if $H$ is symmetric and admits a basis of eigenvalues.
	\end{enumerate}
\end{definition}

Let us briefly comment several facts. If $H$ is symmetric, then the equality $\langle  \psi|H\phi\rangle=\langle H\psi|\phi\rangle$ for $\psi,\phi\in D(H)$ ensures that $H_\psi$ is continuous for every $\psi\in D(H)$ and  $ D(H^\dagger)\supset D(H)$. Hence, $D(H^\dagger)$ is dense in $\mathcal{H}$ and, consequently,  $H^\dagger$ is a well-defined operator. Moreover, 
$H^{\dagger}\big\vert_{D(H)}= H$, where $H^{\dagger}\big\vert_{D(H)}$ is the restriction of $H^\dagger$ to $D(H)$,  and $H^{\dagger}$ is therefore an  extension of $H$. On the other hand, if $H$ is symmetric and $D(H^{\dagger})=D(H)$, then $H^{\dagger}=H$ and vice versa. 
An analogue discussion can be applied to the case of $H$ being a skew-symmetric operator.
\begin{definition}\label{Def::CloOpe}
	Let $A: D(A) \subset \mathcal {X} \rightarrow \mathcal{Y}$ be an operator between normed spaces. We say that $A$ is open if the image of every open set by $A$ is open. The operator $A$ is said to be \textit{closed} if its graph is closed in $\mathcal{X}\times \mathcal{Y}$ relative to the product topology. 
\end{definition}

Sometimes in the literature, a map whose graph is closed is said to be \textit{graph-closed}, and the term \textit{closed} is saved for those operators such that $A(C)$ is closed whenever $C$ is closed. However, we will stick to the terminology established in Definition \ref{Def::CloOpe}.

\begin{theorem} {\bf (Closed graph theorem \cite{Hal_13})} Let  $A: \mathcal{X}\rightarrow \mathcal{Y}$ be an operator between Banach spaces. The operator $A$ is continuous if and only if its graph is closed in $\mathcal{X}\times \mathcal{Y}$.
\end{theorem}

If the $A: \mathcal{X}\rightarrow \mathcal{Y}$ above is continuous, then the closed graph theorem yields that $A$ is closed. Nevertheless, an operator defined on a non-closed domain may be non-continuous and closed simultaneously. For instance, the derivative $\partial:\Cone[0,1]\subset \Czero[0,1]\rightarrow \Czero[0,1]$, where $\Cone[0,1]$ and $\Czero[0,1]$ stand for the spaces of differentiable with continuous derivative and continuous functions  on $[0,1]$, respectively, is closed but not continuous (see \cite{Yos_95}). More generally, every self-adjoint operator is closed \cite{Yos_95}. Then, if $H$ is a symmetric operator such that $D(H)=\mathcal{H}$, then $H$ is bounded. 

\begin{theorem} {\bf (Open mapping theorem \cite[Theorem 4.10]{FHHMZ_11})} If $A:\mathcal{X}\rightarrow \mathcal{Y}$ is a surjective continuous operator between Banach spaces, then $A$ is an open map.
\end{theorem}

\section{Geometry on normed and Banach manifolds}\label{sec:3}

Our geometric study of quantum models will entail the analysis of classes of normed spaces \cite{FB_66,Mar_68}. Since the Banach fixed point theorem does not hold on normed spaces, first-order differential equations in normal form do not need to admit unique integral curves (cf. \cite{Ham_82,KM_97}). There exist no straightforward analogues of the Implicit function and Inverse function theorems neither (see \cite{FB_66,Ham_82} for details). This section recalls the main results on differential calculus on normed spaces to be used hereafter. For details, we refer to \cite{Col_12,FB_66,Ham_82,Mar_68,Nee_05}.

We hereafter assume that every normed vector space $E$ over a field $\Bbbk$ admits an {\it unconditional Schauder basis}, namely a topological basis $\{e_i\}_{i\in \mathbb{N}}$ so that
every $v\in E$ admits a unique decomposition $v=\sum_{i=1}^\infty \lambda^i e_i$ for certain unique constants $\{\lambda^i\}_{i\in \mathbb{N}}\subset\Bbbk$ such that the sum does not depend on the order of summation of its terms. Hereafter, all basis are considered to be unconditional Schauder bases. The field $\Bbbk$ is hereafter assumed to be the field of reals or complex numbers.

Let $E$ and $F$ be normed spaces and let $B(E,F)$ denote the space of continuous linear maps from $E$ to $F$. The space $B(E,F)$ admits an induced topology (for details see \cite{FB_66}). Similarly, $B^r(E,F)$ stands for the space of continuous $r$-linear maps on $E$ taking values in $F$. 

Partial derivatives on normed vector spaces can be defined using the {\it G\^ateaux differential}. In particular, a function $f:U\subset E\rightarrow \mathbb{R}$ on an open $U$ of the normed space $E$ admits a {\it G\^ateaux differential} at $u\in U$ in the direction $v\in E$ if the limit
\begin{equation*}
(\mathfrak{D}_vf)(u):=\lim_{t\rightarrow 0}\frac{f(u+tv)-f(u)}{t}
\end{equation*}
exists. We call the function $u\mapsto (\mathfrak{D}_vf)(u)$ the {\it partial derivative} of $f$ in the direction $v\in E$. If $\mathfrak{D}_vf(u)$ exists for every $v\in E$ at $u\in U$, we say that $f$ is {\it G\^{a}teaux differentiable} at $u$ and we define $(\dd f)_u:v\in E\mapsto \mathfrak{D}_vf(u)\in \mathbb{R}$. This function satisfies that $(\dd f)_u(\lambda v)=\lambda(\dd f)_u(v)$ for every $v\in E$ and $\lambda \in \Bbbk$. Nevertheless, $(\dd f)_u$ does not need to be linear and, even if it is linear, it does not need to be continuous. If $(\dd f)_u$ is linear and continuous, then we say that $f$ is {\it Fr\'{e}chet differentiable} at $u$.  

The function $f:U\subset E \rightarrow \mathbb{R}$ is said to be differentiable if it is Frobenious differentiable at every point of $U$ and the mapping $u\mapsto \dd f_u$ is continuous \cite[pg. 38]{Col_12}. 
The function $f:U\subset E \rightarrow \mathbb{R}$ is said to be differentiable of order $r+1$ at $a\in U$ if it is differentiable of order $r$ and it differential $(d^rf)'$ exists at $a$ 
\cite[Section 4.4]{Col_12}. We say that $f$ is {\it smooth} or of class $\Cinfty$ if it is differentiable for every $r$.

The composite mapping theorem and the Leibniz rule hold on sufficiently differentiable functions on normed spaces \cite{Mar_68}. Although there is no \textit{mean value theorem} in normed spaces, if $(\dd f)_u = 0$ for every $u\in U$ and $U$ is connected, then $f$ is constant on $U$ (see \cite{Col_12,Mar_68}). This result will be employed for studying the existence of Hamiltonian functions related to self-adjoint operators in quantum mechanics (see Section \ref{sec:6}).

Once a notion of differentiability has been stated on normed vector spaces, the definition of a manifold modelled on a normed vector space is similar to the one on a standard manifold modelled on a finite-dimensional vector space (see \cite{Col_12,Con_90,FB_66,Lan_02,Mar_68}).

If not otherwise stated, manifolds are assumed to be smooth and modelled on a separable Banach space $\mathcal{X}$, i.e. the manifold $P$ admits a smooth differentiable structure whose charts take values in $\mathcal{X}$. Exceptions to this general rule appear in our description of quantum mechanical problems related to unbounded operators, which require the use of smooth manifolds modelled on normed spaces that does not need to be complete.

It can be proved that exists a natural one-to-one correspondence between the space $\T_uP$ of {\it tangent vectors} at $u\in P$, which are understood as equivalence classes of curves $\gamma:t\in \mathbb{R}\mapsto \gamma(t)\in P$ passing through $u$ at $t=0$ and having the same, well-defined, G\^ateaux differential at $t=0$ in the direction $1\in \mathbb{R}$ \cite{Lan_02}, and the space $\mathcal{D}_uP$ of derivations $\mathfrak{D}_u:\Cinfty_u(P)\rightarrow \mathbb{R}$, where $\Cinfty_u(P)$ is the space of {\it germs} at $u\in P$ of functions on $P$, of the form $\mathfrak{D}_u f = \frac{\dd}{\dd t}\big\vert_{t=0} (f \circ \gamma)$ for a curve $\gamma$ with $\gamma(0)=u$. Let us prove the analogue of this result on normed spaces under our given general assumptions. The case for general manifolds modelled on normed spaces follows then by using local charts and their compatibility relations. 

Let $(E, \Vert \cdot \Vert)$ be a real normed space. By our general assumption on normed spaces, $E$ admits an unconditional Schauder basis $\{\psi_n\}_{n\in \mathbb{N}}$. Let $\{f^n\}_{n\in \mathbb{N}}$ be the corresponding coordinates in $E^*$. Thus, if $\psi\in E$, then there exist unique real numbers $c^n$, with $n\in\mathbb{N}$, given by $c^n = f^n(\psi)$, such that
\begin{equation*}
\psi = \sum\limits_{n\in \mathbb{N}} c^n \psi_n\,.
\end{equation*}
The elements of $\{f^n\}_{n\in \mathbb{N}}$ do not need to be a basis of the dual space $E^*$, in the sense that certain elements of $E^*$ may not be the limit of a sequence of finite linear combinations of $\{f^n\}_{n\in \mathbb{N}}$ in the strong topology of $E^*$ (consider for instance $\ell^1$ with the basis $e_n=(0,
\ldots,1(n-{\rm th\,\, position}),0,\ldots)$ with $n\in \mathbb{N}$, and its dual $(\ell^1)^*=\ell^\infty$ with the elements $\sum_{n\in \mathbb{N}}f^n\in \ell^\infty$). Anyhow, the value of any element of $E^*$ at an element $v\in E$ can be obtained as a linear combination  $\sum_{n\in \mathbb{N}}\lambda_nf^n(v)$ where the sum is understood as a limit in the field $\mathbb{K}$. In general, the $\{\lambda_n\}_{n\in \mathbb{N}}$ satisfy relations ensuring that the sum exists and it is continuous. These relations can be different from the ones satisfied by the $\{c^n\}_{n\in \mathbb{N}}$, e.g. the space $\ell^p$ of sequences ${\bf c}=(c^1,c^2,\ldots)$ of numbers with $\sum_{n\in \mathbb{N}}|c^n|^p<\infty$ with respect to the norm $\|{\bf c}\|=\left(\sum_{n\in \mathbb{N}}|c^n|^p\right)^{1/n}$ satisfies that $\ell^{p*}\simeq \ell^{q}$ for $1/p+1/q=1$. 

Every equivalence class of curves passing through $\psi_0\in E$ is determined by the G\^ateaux differential of the curves at $t=0$ in the direction $1\in \mathbb{R}$ and vice versa. Then, there is a one-to-one correspondence between $\T_{\psi_0}E$ and $E$, and we can write $\T_{\psi_0}E\simeq E$.

In turn, $\psi\in E$ determines a derivation 
\begin{equation}\label{eq:psidot}
\dot{\psi}_{\psi_0}f =\frac{\dd}{\dd t}\bigg\vert_{t=0}\big( f(\psi_0+t\psi)\big)\,,\qquad \forall f\in \Cinfty(E)\,.
\end{equation}
Obviously $\dot\psi_{\psi_0}\in \mathcal{D}_{\psi_0}E$. The subindex $\psi_0$ will be omitted if it is clear from context, or irrelevant, so as to simplify the notation.

If $\dot{\psi}_{\psi_0} = \dot{\phi}_{\psi_0}$, then it follows from \eqref{eq:psidot} that $f^n(\psi) =f^n(\phi)$ for all $n\in \mathbb{N}$ and $\psi=\phi$. Then, every element of $\mathcal{D}_{\psi_0}E$ determines a unique element of $E$, which defines a unique tangent vector at $\psi_0$. Then, we can identify $\T_{\psi_0}E$ and $\mathcal{D}_{\psi_0}E$ and write $\T_{\psi_0}E := \{ \dot{\psi}_{\psi_0}  \mid \psi \in E \}.$ Moreover, each $\psi\in E$ is associated with a derivation $\dot{\psi}\in \T_{\psi_0}E$ via the isomorphism  $\lambda_{\psi_0} :\psi\in E \mapsto \dot{\psi}_{\psi_0}\in \T_{\psi_0}E$.

The previous construction allows for building the tangent and cotangent bundles of a normed manifold $P$ along with other geometric structures, e.g. generalised tangent bundles \cite{Lee_12}, following the same ideas as in the finite-dimensional case, provided the infinite-dimensional nature of these structures is taken into account. In particular, a \textit{vector field} is a section of the tangent bundle $\T P$, more precisely, it is a map $X: P \rightarrow \T P$ such that $\tau \circ X = \Id_P$, where $\tau: \T P \rightarrow P$ is the projection onto the base space $P$ and $\Id_P$ is the identity map on $P$ (see \cite{Mar_68}) Similarly, a {\it  $k$-form} is a section of the bundle $\bigwedge^k \cT P\rightarrow P$, where $(\bigwedge^k \cT P)_u$ is the space of continuous $\mathbb{R}$-valued $k$-linear skew-symmetric  mappings on $\T_uP$. We hereafter write $\Omega^k(P)$ for the space of differential $k$-forms on $P$ and we denote  by $\Cinfty(P)$ the space of smooth functions on $P$. Similarly to the finite-dimensional manifold case, one can define the \textit{exterior derivative} $\dd: \Omega^k(P)\rightarrow \Omega^{k+1}(P)$, which exists on an appropriate subset of $\Omega^k(P)$.

If $P$ is not modelled on a Banach space, then a vector field on $P$ may not have certain integral curves \cite{KM_97}. Additional conditions must be required to ensure the existence of integral curves. Moreover, our study of quantum mechanical problems will demand the definition of vector fields with a domain (see \cite{Mar_68}) for details).  
\begin{definition}
A \textit{vector field on a domain} $D\subset P$ is a map $X: D\subset P \rightarrow \T P$ such that $D$ is a manifold modelled on a normed space that, as a subset of $P$, is dense in $P$,  the inclusion $\iota:D\hookrightarrow P$ is smooth, and $\tau\circ X = \Id_D$. Moreover, $X$ is called {\it smooth} if $X:D\rightarrow \T P$ is smooth.
\end{definition}

\begin{example}
\label{ex:partialxfield}
Let us provide an example of a vector field on a domain that will be generalised and studied in detail in Section \ref{sec:6}. Consider the complex Hilbert space $L^2(\mathbb{R})$ of equivalence classes of square integrable complex functions on $\mathbb{R}$ that coincide almost everywhere. This space is naturally a real manifold modelled on a real Hilbert space given by the differential structure induced by the global map  $\jmath:L^2(\mathbb{R})\rightarrow L_\mathbb{R}^2(\mathbb{R})$, where $L^2_{\mathbb{R}}(\mathbb{R})$ is $L^2(\mathbb{R})$ considered, in the natural way, as a real Hilbert space. The isomorphisms $\lambda_{\psi}: L^2(\mathbb{R})\rightarrow \T_{\psi}L^2(\mathbb{R})$, for every $\psi \in L^2(\mathbb{R})$, allow us to write $\T L^2(\mathbb{R})\simeq L^2(\mathbb{R})\oplus L^2(\mathbb{R})$ and to consider every operator 
$$A:D(A)\subset L^2(\mathbb{R})\rightarrow L^2(\mathbb{R}),$$ as a map 
\begin{equation}\label{eq:VecOpe}X_A:\psi\in D(A)\subset L^2(\mathbb{R})\mapsto (\psi, A\psi)\in \T_\psi L^2(\mathbb{R}).
\end{equation}

In particular, consider $\partial_x:D(\partial_x)\subset L^2(\mathbb{R})\rightarrow L^2(\mathbb{R})$ whose $D(\partial_x)$ is considered to be spanned by linear combinations of a finite family of elements of the basis of functions $H_n(x)e^{-x^2/2}/(\pi^{1/4}2^{n/2}\sqrt{n!})$, with $n\in\bar{\mathbb{N}} = \mathbb{N}\cup \{0\}$ and $H_n(x)$ being the Hermite polynomial of order $n$. Then, $D(\partial_x)$ is dense in $L^2(\mathbb{R})$ (cf. \cite{Sze_39}). Moreover, $D(\partial_x)$ is a normed vector space relative to the restriction of the topology of $L^2(\mathbb{R})$ to $D(\partial_x)$, but it is not complete. The restriction of the global atlas of $L^2(\mathbb{R})$ to $D(\partial_x)$ induces a second one on $D(\partial_x)$ that is modelled only on a normed, not complete, space because $L^2(\mathbb{R})=\overline{D(\partial_x)}\neq D(\partial_x)$. Note that $D(\partial_x)\simeq \bigoplus_{n
\in 
\mathbb{N}}\mathbb{R}$, where one assigns every element of $D(\partial_x)$ its coordinates in the given basis. Then, the inclusion of $D(\partial_x)$ into $L^2(\mathbb{R})$ reads in these coordinates as
$$
D(\partial_x)\ni(x_1,x_2,x_3,\ldots)\longmapsto (x_1,x_2,\ldots)\in L^2(\mathbb{R})
$$
and becomes smooth. In fact, the coordinates of an element in $D(\partial_x)$ may be any finite sequence of non-zero constants, while any element of $L^2(\mathbb{R})$ is a square summable sequence. Evidently, every smooth function on $L^2(\mathbb{R})$ is smooth on $D(\partial_x)$.  
 Since $\partial_x (D(\partial_x))\subset D(\partial_x)$, we can write that $X_{\partial_x}: D(\partial_x) \rightarrow \T D(\partial_x)$. In coordinates adapted to the basis $H_n^p(x)e^{-x^2/2}$, where $H_n^p(x)$ is the Hermite probabilistic polynomial of order $n$, and using that $\partial_xH_n^p(x)-xH_n^p(x)=-H_{n+1}^p(x)$ for any non-negative integer $n$, one has
$$
X_{\partial_x}(x_1,x_2,x_3,\ldots)=(x_1,x_2,x_3,\ldots;0,-x_1,-x_2,-x_3,\ldots)\in \T D(\partial_x).
$$
Hence, $X_{\partial_x}$ becomes a vector field on a domain.
\demo 
\end{example}

Example \ref{ex:partialxfield} can be generalised to vector fields on domains induced by observables, which will be the standard case in our study of quantum mechanics. This case will also turn out to be very frequent in the literature \cite{Mar_68}.

Similarly to vector fields on a domain, one can define differential $k$-forms on a domain. We write $\Omega^k(D)$ for the space of differential $k$-forms on a domain $D$. Moreover, it is straightforward to define $t$-dependent vector fields on a domain \cite{AM_78}.

\begin{definition}
A \textit{(local) flow} for $X: D \subset P\rightarrow \T P$ is a  map of the form $F: (t,u)\in (-\varepsilon, \varepsilon ) \times U \mapsto F_t(u) := F(t,u)\in P$, where $\varepsilon>0$ and $U\subset P$ is open and contains $D$, such that
\begin{equation*}
X(F_t(u))=\frac{\dd}{\dd s}\bigg|_{s=t} F_s (u)\,,
\end{equation*}
for all $u\in D$ and $t\in (-\epsilon,\epsilon)$. 

\end{definition}

At first glance, it may seem strange that the flow of a $t$-dependent vector field $X$ on a domain $D$ may be defined on an open set $U$ strictly containing $D$, but this will be useful and natural in this and others works (see \cite{Mar_68}). More precisely, skew-self-adjoint operators on a Hilbert space will give rise to vector fields on domains whose flows are defined on the full $\mathbb{R}\times \mathcal{H}$ but they are differentiable only on a subspace of it. 

\begin{example}
    Recall the vector field $X_{\partial_x}$ from Example \ref{ex:partialxfield}. For an element $\psi\in D(\partial_x)$,  since $\partial_x \psi\in D(\partial_x)$ and $\psi$ is an analytic function, one has that $[\exp(t\partial_x)\psi](y)=\sum_{n=0}^\infty \frac{t^n\partial^n_x\psi}{n!}(y)=\psi(y+t)$ and $\exp(t\partial_x)\psi$ may not to belong to $D(\partial_x)$ (cf. \cite{Hal_13}). Since $\ii \partial_x$ admits a self-adjoint extension, the Stone--von Neumann theorem ensures that $\exp(t\partial_x)$ can be defined as a unitary map on $L^2(\mathbb{R})$ and $X_{\partial_x}$ admits a flow given by $F:(t;\psi)\in \mathbb{R}\times L^2(\mathbb{R})\mapsto \exp(t\partial_x) \psi \in L^2(\mathbb{R})$. For elements $\psi\in D(\partial_x)$, one observes that $\frac{\d}{\d t} F(t, \psi) = \partial_x \psi$, but $F$ may not be differentiable for $\psi \notin D(\partial_x)$. 
    
    Note that $D(\partial_x)$ can be chosen so as to give rise to a flow $F$ of $X_{\partial_x}$ on a domain that is invariant under the maps $F_t$ for $t\in \mathbb{R}$. For instance, consider $\partial_x$ to be defined on the Schwartz space $\mathcal{S}(\mathbb{R})$. It can be proved that $(\exp(t\partial_x)\psi](x)=\psi(x+t)$ for every $\psi\in \mathcal{S}(\mathbb{R})$ and $x\in \mathbb{R}$. Hence, the new $D(\partial_x)$ is invariant relative to $\exp(t\partial_x)$. Moreover, $\frac{\d}{\d t} F(t, \psi) = \partial_x \psi$ on elements of $\mathcal{S}(\mathbb{R})$.\demo
\end{example}

Since the sum and composition of vector fields related to operators with different domains do not need to be defined at a new common 
 domain, the Lie bracket of such vector fields may not be defined. In this work, conditions will be given to ensure that the Lie bracket of studied vector fields is defined on a dense subspace of $P$. This happens, for instance, for observables having a common basis of eigenvalues on the space they are defined on. Then, one can consider the direct sum of the subspaces generated by the elements of such a basis. This is a normed, but generally not complete, subspace. 


It is interesting to our purposes to consider the case when a subset $D$ of a manifold $P$ modelled on a Banach space is a manifold modelled on a normed space. In this case, the  exterior differential operator $\dd$ can be similarly defined on differential forms on $D$. Moreover, $\T_p D\subset \T_pP$ and one can consider whether a function $f:D\rightarrow \mathbb{R}$ is such that its differential, namely ${\dd f}_u:\T_uD\rightarrow \mathbb{R}$, can be extended at points $u\in D$ to elements of $(\dd f)_u\in \cT_pP$. This will be important in our quantum mechanical applications and it becomes the main reason to provide the following definition.

\begin{definition} Let $D$ be a subset of a manifold $P$ that is, again, a manifold modelled on a normed space. We say that $f: D \rightarrow \mathbb{R}$ is \textit{admissible} when $(\dd f)_{u}: \T_{u} D \rightarrow \mathbb{R}$ can be extended continuously to $\T_{u}P$.
\end{definition}

For a vector field $X$ on a domain $D_X \subset P$, which is assumed to be a manifold modelled on a normed space such that  $X(D_X)\subset \T D_X$, the \textit{inner contraction},  $\inn{X}$, on differential forms with $D_X\supset D_\alpha$ is defined by restricting the contraction with $X$ of differential forms at points of $P$ and its evaluation on vector fields on $D_X$.

Similarly, the \textit{Lie derivative} $\Lie_X:\alpha\in \Omega^k(P) \mapsto \Lie_X\alpha\in \Omega^k(D_X)$ is defined using Cartan's formula:
\begin{equation}\label{eq:Liederivative}
    (\Lie_X \alpha)(u) = (\inn{X} \dd \alpha +\dd \inn{X} \alpha) (u)\,,\qquad \forall u\in D_X\,.
\end{equation}
\begin{theorem}
    Let $X$ be a vector field on a domain $D_X\subset P$ and let  $F: (-\epsilon,\epsilon)\times D_X \rightarrow D_X$ be the flow of $X$. Then,
    \begin{equation*}
        [F_t^*(\Lie_X \alpha)]_u = \frac{\dd}{ \dd s}\bigg|_{s=t} (F_s^*\alpha)_u\,,
    \end{equation*}
    for any $\alpha \in \Omega^k(P)$, values $t,s\in (-\epsilon,\epsilon)$, and arbitrary $u\in D_X$. 
\end{theorem}

For a diffeomorphism $\phi: P_1\rightarrow P_2$ and a vector field $X$ on a domain $D_X\subseteq P_1$, the \textit{push-forward} of $X$ by $\phi$, if it exists, has domain $\phi(D_X)\subseteq P_2$ and is defined by $(\phi_* X)_u = \T_u\phi [X_{\phi^{-1}(u)}]$. If the flow of $X$ is given by $F_t$, then the flow of $\phi_* X$ is $\hat F_t:=\phi\circ F_t\circ \phi^{-1}$ for every $t\in (-\epsilon,\epsilon)$. In particular, one has to assume that $\phi(D_X)$ is dense in $P_2$ for the vector field $\phi_*X$ to exist. 

\begin{definition}
A diffeomorphism $\phi: D\rightarrow D$ of a domain $D\subset P$ is said to be \textit{admissible} if and only if for every $u\in D$, the map $\T_u\phi:\T_uD\rightarrow \T_{\phi(u)}D$ extends to a continuous linear map $\T_uP\to\T_{\phi(u)}P$. 
\end{definition}

The following results are the generalisation to the infinite-dimensional context of the quotient manifold theorem \cite{AMR_88} (see also \cite{Bou_89,Bou_07} for details).

\begin{theorem}\label{Iloc}
    Let $\Phi:G\times P\rightarrow P$ be a free and proper action on a Banach manifold. Then, $P/G$ is a manifold and $\pi:P\rightarrow P/G$ is a submersion.
\end{theorem}
\begin{proposition}\label{prop:compactLiegroup}
    If $G$ is a compact Lie group, then every Lie group action $\Phi:G\times P\rightarrow P$ is proper.
\end{proposition}

\section{Symplectic geometry in the infinite-dimensional context}\label{sec:4}

This section introduces and studies symplectic and related structures on manifolds modelled on normed and Banach spaces. As shown next, the fact that a symplectic form is defined on such, generally infinite-dimensional, manifolds yields several relevant differences with respect to the case of symplectic forms on finite-dimensional manifolds. 

\begin{definition}
    A \textit{symplectic form} is a closed differential two-form $\omega$ on a manifold $P$ modelled on a normed space that is {\it non-degenerate}, namely $\omega^\flat_u: v_u\in \T_uP  \mapsto \omega_u(v_u,\cdot)\in \cT_uP$ is injective  for every $u\in P$. If $\omega_u^{\flat}$ is an isomorphism for every $u\in P$, then we call $\omega$ a \textit{strong symplectic form}. A \textit{(strong) symplectic manifold} is a pair $(P, \omega)$, where $\omega$ is a (strong) symplectic form on $P$.
\end{definition}

Every strong symplectic form is always a symplectic form, whereas a symplectic form $\omega$ on $P$ is strong provided $\omega^\flat_u$ is also surjective for every $u\in P$. The latter can be warranted, for instance, when $P$ is finite-dimensional. If the symplectic form $\omega$ of a  symplectic manifold $(P,\omega)$ is known from context, then one will simply say that $P$ is a symplectic manifold.

The Darboux theorem is a well-known result from standard finite-dimensional symplectic geometry that states that every symplectic manifold admits a local coordinate system in which the associated symplectic form takes a canonical form \cite{AM_78,Lee_12}. This result cannot be extended to symplectic forms on manifolds modelled on normed spaces, as shown by Marsden for the case of Banach symplectic manifolds \cite{Mar_68}, which in turn  led to different generalisations of the Darboux Theorem for Banach symplectic manifolds \cite{Bam_99,Mar_81,Tro_76}. 

The extension of the standard Darboux theorem to Banach symplectic manifolds is accomplished through {\it constant differential forms}. Let us describe this notion. Every skew-symmetric continuous bilinear form on a normed space $E$, let us say $B:E\times E\rightarrow \mathbb{R}$, gives rise to a unique differential two-form $\omega_0$ on $E$ by considering the natural isomorphism $\lambda_u:E\simeq \T_uE$ and defining $\omega_0$ to be the only differential form satisfying that $\lambda_u^*(\omega_0)_u= B$ for every $u\in E$. 
Such a type of differential two-form on a normed space $E$ defined from of a skew-symmetric bilinear form on $E$ is called a {\it constant} differential two-form. 

\begin{theorem}{{\bf (Strong Darboux theorem \cite{AM_78})}}
    If $(P,\omega)$ is a strong Banach symplectic manifold, then there exists an open neighbourhood around every $u\in P$ where $\omega$ is the pull-back of a constant differential two-form.
\end{theorem}

The cotangent bundle of any finite-dimensional manifold admits a so-called {\it tautological differential one-form} whose differential gives rise to a strong symplectic form \cite{AMR_88}. Similarly, the cotangent bundle of a Banach manifold $P$ is endowed with a canonical differential one-form and a related symplectic form, which is not strong in general \cite{Mar_68}. 

\begin{definition}
    Let $P$ be a manifold modelled on a Banach space $\mathcal{X}$ and let $\pi:\cT P\rightarrow P$ be the cotangent bundle projection. The \textit{tautological differential one-form} on $\cT P$ is the differential one-form $\theta$ on $\cT P$ of the form
    \begin{equation*}
        \theta_{\alpha_u}(w_{\alpha_u}) :=  \alpha_u(\T\pi(w_{\alpha_m}))\,,\qquad \forall\alpha_u\in \cT_uP\,, \quad \forall w_{\alpha_u}\in \T_{\alpha_u}(\cT P)\,, \quad \forall u\in P\,,
    \end{equation*}
    and the \textit{canonical two-form} $\omega$ on $P$ is defined as $\omega=-\dd \theta$. 
\end{definition}

Let us give a coordinate expression for $\theta$ and $\omega$ on a separable Banach manifold $P$ modelled on a Banach space $\mathcal{X}$ with an unconditional Schauder basis. It is worth recalling that the separability does not imply the existence of an unconditional Schauder basis (see \cite{Enf_73}). On a sufficiently small open set $U\subset P$, one can consider a coordinate system $\{x^j\}_{j\in \mathbb{N}}$. The coordinate system $\{x^j,p_j\}_{j\in\mathbb{N}}$ adapted to $\cT P$ allows us to establish a local diffeomorphism $\cT U\simeq U\times \mathcal{X}^*$. In turn, the coordinate system $\{x^j,p_j\}_{j\in\mathbb{N}}$ induces another adapted coordinate system to $\T\cT U$, namely $\{x^j,p_j,\dot x^j,\dot p_j\}_{j\in\mathbb{N}}$, allowing us to define a diffeomorphism $\T\cT U\simeq (U\times \mathcal{X}^*)\times (\mathcal{X}\times \mathcal{X}^*)$. In these coordinates,
\begin{equation}\label{eq:LocalThe}
    \sum\limits_{j\in\mathbb{N}}\theta_{(x,p)}(\dot x^j(x,p)\partial_{x^j}+\dot p_j(x,p)\partial_{p_j})=\sum\limits_{j\in\mathbb{N}}p_j\dot x^j \quad\Longrightarrow \quad\theta=\sum\limits_{j\in\mathbb{N}}p_j\dd x^j.
\end{equation}
Then, $\omega=-\dd \theta = \sum_{j\in \mathbb{N}} \dd x_j \wedge \dd p^j$.

More generally, one has the following result for Banach spaces (see \cite{Mar_68} for details).
\begin{theorem}\label{RefCan}
    The canonical two-form $\omega$ on $\cT P$ is a symplectic form. In addition, $\omega$ is a strong symplectic form if and only if $P$ is modelled on a reflexive Banach space.
\end{theorem}

\begin{definition}
    Let $P$ be a Banach manifold and let $\{\cdot, \cdot\}$ denote a bilinear operation on $\Cinfty(P)$. If $(\Cinfty(P), \{\cdot,\cdot\})$ is a Lie algebra and $\{\cdot, \cdot\}$ satisfies the {\it Leibniz rule}, namely $\{fg, h\} = \{f,h\}g + f\{g,h\}$ for all $f,g,h\in \Cinfty(P)$, then $\{ \cdot, \cdot\}$ is called a \textit{Poisson bracket} or a \textit{Poisson structure} on $P$, and in this case $(P, \{\cdot, \cdot \})$ is said to be a \textit{Poisson manifold}.
\end{definition}

\begin{definition}\label{def:ham_vecfield}
    Let $(P,\omega)$ be a symplectic manifold.  A vector field $X: P \rightarrow \T P$ is called \textit{Hamiltonian} if there exists  a function $h\in \Cinfty(P)$ such that $\inn{X}\omega=\dd h$. The function $h$ is called a \textit{Hamiltonian function} of $X$. Whenever $\inn{X}\omega$ is closed, and thus locally exact (see \cite{Mar_68,MW_74}), we say that $X$ is a \textit{locally Hamiltonian vector field}.
\end{definition}

The set of (locally) Hamiltonian vector fields on $(P,\omega)$ will be denoted by (resp. $\Hamloc(P,\omega)$) $\Ham(P,\omega)$ or (resp.  $\Hamloc(P)$) $\Ham(P)$ when $\omega$ is understood from context.

If $\omega$ is a strong symplectic form, then every $h\in \Cinfty(P)$ is the Hamiltonian function of a Hamiltonian vector field as in the case of symplectic forms on finite-dimensional manifolds. For general Banach symplectic manifolds, each $\omega_u^{\flat}$, with $u\in P$, is only in general injective and some $h\in \Cinfty(P)$ may not give rise to any Hamiltonian vector field. If a function $h\in\Cinfty(P)$ has a Hamiltonian vector field, it is unique. This motivates the following definition.

\begin{definition}
    Let $(P,\omega)$ be a Banach symplectic manifold. We call a function $h\in\Cinfty(P)$ {\it admissible} if it is the Hamiltonian function of a Hamiltonian vector field.
\end{definition}

We will denote by $X_h$ the unique Hamiltonian vector field associated with an admissible function $h$. The following result is immediate. 

\begin{proposition}\label{Prop:LieHam}
    Let $(P,\omega)$ be a symplectic manifold. The set $\Hamloc(P)$ is a Lie algebra, while $\Ham(P)$ is an ideal of $\Hamloc(P)$ and $[\Hamloc(P),\Hamloc(P)]\subset \Ham(P)$. In particular, if $X_1,X_2\in \Hamloc(P)$, then $[X_1,X_2]=-X_{\omega(X_1,X_2)}$.
\end{proposition}

Every symplectic form on a finite-dimensional manifold gives rise to a Poisson structure. This is no longer true  for a general symplectic form, as not every function is associated with a Hamiltonian vector field. A way to ensure this is given by the following proposition.
\begin{proposition}
    A strong symplectic manifold $(P,\omega)$ induces a Poisson structure on $P$ given by 
    \begin{equation}\label{PoIn}
        \{f,g\}_\omega := \omega( X_f, X_g)\,, \qquad \forall f,g \in \Cinfty(P)\,.
    \end{equation}
\end{proposition}

\begin{definition}
    Let $(P_1,\omega_1)$ and $(P_2,\omega_2)$ be two symplectic manifolds. A diffeomorphism $\phi : P_1 \rightarrow P_2$ is called a \textit{symplectic map}, or a \textit{symplectomorphism}, if $\phi^*\omega_2=\omega_1$.
\end{definition}

We hereafter focus on physical systems determined by the so-called Hamiltonian systems.

\begin{definition}
    A \textit{Hamiltonian system} is a triple $(P,\omega, h)$, where $(P,\omega)$ is a symplectic manifold and $h:(t,x)\in \mathbb{R}\times P\mapsto h_t(x)=h(t,x)\in \mathbb{R}$ is a $t$-dependent function on $P$ such that for every $t\in\R$, the function $h_t\in\Cinfty(P)$ is an admissible function for $(P,\omega)$. We call $h$ the \textit{Hamiltonian function} of the system.
\end{definition}

Every Hamiltonian system $(P,\omega,h)$ gives rise to a unique $t$-dependent vector field $X$, called the \textit{Hamiltonian vector field} of the system, satisfying that $\inn{X_t}\omega = \dd h_t$. The system of differential equations describing the integral curves $c:t\in \mathbb{R}\mapsto c(t)\in P$ of $X$ reads 
\begin{equation*}
    \frac{\dd c(t)}{\dd t} = (X_t\circ c)(t)\,.
\end{equation*}
This system of differential equations is called the \textit{Hamilton's equations} for the system $(P,\omega, h)$. To emphasise the role of the Hamiltonian function $h$ to determine its associated vector field, we will sometimes write $X_h = X$. In the Marsden--Weinstein reduction procedure, one is interested in a special family of Hamiltonian systems admitting a particular type of symmetries described next.

\begin{definition}
    A \textit{$G$-invariant Hamiltonian system} relative to the Lie group action $\Phi:G\times P\rightarrow P$ is a Hamiltonian system $(P,\omega,h)$ such that
    \begin{equation*}
        \Phi_g^*h_t = h_t\,,\qquad \Phi^*_g\omega=\omega\,,\qquad \forall g\in G\,,\qquad \forall t\in \mathbb{R}\,.
    \end{equation*}
\end{definition}
If $\Phi:G\times P\rightarrow P$ is understood from context, we will simply talk about a $G$-invariant Hamiltonian systems, omitting the Lie group action.

\begin{definition}
    A Lie group action $\Phi: G\times P \rightarrow P$ on a symplectic manifold $(P, \omega)$ is called \textit{symplectic} if, for every $g\in G$, the map $\Phi_g:P\rightarrow P$ is a symplectomorphism. Meanwhile, $\Phi$ is called {\it Hamiltonian} if the fundamental vector fields of $\Phi$ are Hamiltonian relative to the symplectic form $\omega$.
\end{definition}

\begin{example}
    \label{Ex:expaction}
    Let $\mathcal{H}$ be a complex Hilbert space and consider the constant symplectic form $\omega_B$ on $\mathcal{H}$ induced by the continuous skew-symmetric bilinear form $B:(\psi_1,\psi_2)\in \mathcal{H}\times \mathcal{H}\mapsto \mathfrak{Im}\langle \psi_1|\psi_2\rangle \in\mathbb{R}$. Given the skew-adjoint operator $\ii\Id_\mathcal{H}$, the Stone--von Neumann  Theorem \cite{Hal_13,Sto_32} states that each $\exp(\ii t\Id_{\mathcal{H}})$ is a continuous mapping on $\mathcal{H}$ and, since it is linear, it is also smooth (cf. \cite{Nee_05}). Since $t\ii\Id_{\mathcal{H}}$ is continuous, $\exp(\ii t\Id_{\mathcal{H}})=\ee^{\ii t}\Id_{\mathcal{H}}$ and the operators $\exp(\ii t\Id_{\mathcal{H}})$, for $t\in \mathbb{R}$, span a group isomorphic to $U(1)$, namely the unitary group of transformations on a one-dimensional Hilbert space. Define the Lie group action
    \begin{equation*}
        \Phi:(\exp(t{\rm i}\Id_\mathcal{H}),\psi)\in U(1)\times \mathcal{H}\longmapsto e^{\ii t}\psi\in \mathcal{H}\,.
    \end{equation*}
    Then, $B(\Phi_g\psi_1,\Phi_g\psi_2)=B(\psi_1,\psi_2)$ for every $\psi_1,\psi_2\in \mathcal{H}$ and $g\in U(1)$. Consequently,  $\Phi^*_g\omega_B=\omega_B$ for every $g\in G$ and $\Phi$ is a symplectic action.\demo
\end{example}

\begin{definition}
Let $(P,\omega)$ be a symplectic manifold. A \textit{momentum map} for a Lie group action $\Phi:G\times P\rightarrow P$  is a mapping $J : P \rightarrow \mathfrak{g}^*$ such that
\begin{equation*}
\inn{\xi_P}\omega = \dd \langle J , \xi\rangle\,,\qquad \forall \xi\in \mathfrak{g}\,, \quad \forall u\in P\,,
\end{equation*}
where $\xi_P$ is the fundamental vector field of $\Phi$ related to $\xi\in\mathfrak{g}$, namely
\begin{equation*}
\xi_P(u):=\frac{\rm d}{{\rm d}t}\bigg|_{t=0}\Phi(\exp(t\xi),u)\,,\qquad \forall u\in P\,.
\end{equation*}
\end{definition}

Thus, $J_{\xi}:u\in P\mapsto \langle J(u), \xi \rangle\in \mathbb{R}$ is a Hamiltonian function of $\xi_P$ for each $\xi\in \mathfrak{g}$. 

\begin{theorem}\label{Th:HNT}{\bf(Hamiltonian Noether's theorem \cite{MW_74})} \label{NoetherTh}
    Consider a Hamiltonian system $(P,\omega,h)$. Let $\Phi: G \times P \rightarrow P$ be a Lie group action on a symplectic manifold $(P,\omega)$ with a momentum map $J:P \rightarrow \mathfrak{g}^*$. If the $t$-dependent function $h$ on $P$ is $G$-invariant, namely $h_t\circ \Phi_g = h_t$ for all $g\in G$ and $t\in \mathbb{R}$, then $J$ is conserved on the trajectories of $X_h$, that is $J\circ F_t = J$, for all $t\in \mathbb{R}$, where $F$ is the flow of $X_h$. 
\end{theorem}

To introduce the symplectic Marsden--Weinstein reduction, one may consider a quite general property of momentum maps given by the following definition (cf. \cite{AM_78}), which makes use of the coadjoint action $\Ad:G\times \mathfrak{g}^*\rightarrow \mathfrak{g}^*$ on the dual to the Lie algebra $\mathfrak{g}^*$ of $G$.
\vskip 0.3cm

\hspace*{-\parindent}%
\begin{minipage}{0.75\textwidth}
    \begin{definition}
    A momentum map $J:P\rightarrow \mathfrak{g}^*$ of a Lie group action $\Phi: G\times P\rightarrow P$ is said to be Ad$^*$-\textit{equivariant} when $J\circ \Phi_g = (\Ad^*)_g\circ J$ for all $g\in G$, where $\Ad^*$ is the coadjoint action of $G$, namely  $\Ad^*:(g,\theta)\in G\times \mathfrak{g}^*\mapsto (\Ad^*)_g:=\theta\circ \Ad_{g^{-1}}\in \mathfrak{g}^*$ and $(\Ad^*)_{g}$ is the transpose of the morphism $\Ad_{g^{-1}}:\mathfrak{g}\rightarrow \mathfrak{g}$. In other words, $J$ is $\Ad^*$-\textit{equivariant} if the  diagram aside commutes for every $g\in G$.
    \end{definition}  
\end{minipage}
\hfill
\begin{minipage}{0.2\textwidth}
\begin{tikzcd}
    P
    \arrow[r,"J"]
    \arrow[d, swap, "\Phi_g"]& \mathfrak{g}^*
    \arrow[d,"(\Ad^{*})_g"]\\
    P
    \arrow[r,"J"]&
    \mathfrak{g}^*
    \end{tikzcd}
\end{minipage}

\vskip 0.3cm

The $\Ad^*$-equivariance of momentum maps is not really necessary in what follows: it is only commonly used. Instead, one can use that  a momentum map related to a Lie group action $\Phi:G\times P\rightarrow P$ is always equivariant related to a new Lie group action $\Delta:G\times \mathfrak{g}^*\rightarrow \mathfrak{g}^*$ defined by $\Delta(g,\vartheta)=({\rm Ad}^*)_{g}+
\sigma(g)$ for every $g\in G$ and $\vartheta\in \mathfrak{g}^*$. Moreover, $\sigma:G\rightarrow \mathfrak{g}^*$ is defined to be the well defined constant element in $\mathfrak{g}^*$ of the form
$$
\sigma(g)=J\circ \Phi_g-({\rm Ad}^*)_g\circ J\,,\qquad \forall g\in G\,.
$$
The Lie group action $\Delta$ is frequently called the affine Lie group action related to $J$ and $\Phi$. We hereafter write $G_\mu$ for the isotropy subgroup of $G$ leaving invariant an element $\mu\in \mathfrak{g}^*$ relative to its associated affine Lie group action.

When $(P,\omega,h)$ is a $G$-invariant Hamiltonian system, one expects that the $G$-invariance of the system will allow for the simplification of the system, e.g. giving rise to a new symplectic manifold of `smaller' dimension. This is next accomplished through the Marsden--Weinstein reduction theorem allowing one to go from a symplectic manifold having a Hamiltonian Lie group action of symmetries of $h$ to another `reduced' symplectic manifold. Then, the initial Hamiltonian function $h$ will also give rise to a new $t$-dependent function on the `reduced' manifold, which finishes the reduction procedure (see \cite{MW_74} for details). To do so, we first need the give the following definition.

\begin{definition}
    Given a map $f:P\rightarrow N$ between two manifolds, we say that $n\in N$ is a {\it  weakly regular value} of $f$ if $f^{-1}(n)\subset P$ is a submanifold, $\T_uf^{-1}(n)=\ker \T_{u}f$ for every $u\in f^{-1}(n)$, and $\ker \T_{u}f$ admits a closed supplementary in $\T_uP$.
\end{definition}

Note that the condition on the closed supplementary in $\T_uP$ is immediately satisfied when $P$ is a finite-dimensional manifold. Nevertheless, this property is not always satisfied in the infinite-dimensional case. In the infinite-dimensional real, immersions, submersions and other types of mappings are required to satisfy similar conditions to extend to the infinite-dimensional case the properties they have on finite-dimensional manifolds. 

\begin{theorem}{{\bf (Marsden--Weinstein reduction theorem)}}\label{MWtheorem}
    Consider a symplectic manifold $(P,\omega)$. Let $G$ be a Lie group acting on $P$ and let $J: P \rightarrow \mathfrak{g}^*$ be an associated momentum map. Let $\mu$ be a weakly regular value of $J$ and suppose that $G_{\mu}$ acts freely and properly on the submanifold $J^{-1}(\mu)$. If $\jmath_{\mu}: J^{-1}(\mu) \hookrightarrow P$ is the natural embedding of $J^{-1}(\mu)$ within $P$ and $\pi_{\mu}: J^{-1}(\mu) \rightarrow P_{\mu}$ is the canonical projection, then there exists a unique symplectic form $\omega_{\mu}$ on $P_{\mu}:=J^{-1}(\mu)/G_\mu$ such that 
\begin{equation*}
\pi^*_{\mu} \omega_{\mu}=\jmath^*_{\mu}\omega\,.
\end{equation*}
\end{theorem}
It is important to point out that Theorem \ref{MWtheorem} does not require the Banach symplectic manifold $P$ to be strongly symplectic \cite{MW_74}. The following theorem uses the Marsden--Weinstein reduction theorem to reduce the dynamics.

\begin{theorem}
    Let the conditions of Theorem \ref{MWtheorem} hold. Let $(P,\omega,h_t)$ be a $G$-invariant Hamiltonian system relative to the Lie group action $\Phi:G\times P\rightarrow P$. Then,
    \begin{itemize}
        \item the flow of $X_h$, namely $F$, leaves $J^{-1}(\mu)$ invariant,
        \item the diffeomorphisms $F_t$ commute with every $\Phi_{g}$, for $g\in G_\mu$, giving rise to a flow $K$ on $P_\mu$ such that $\pi_\mu\circ F_t=K_t\circ \pi_\mu$,
        \item the flow $K$ corresponds to a vector field with a Hamiltonian function $h_\mu\in \Cinfty(P_\mu)$ determined by $h_{\mu}\circ \pi_{\mu} = h\circ \jmath_{\mu}$.
    \end{itemize}
\end{theorem}

It is remarkable that the proper description of the theorem in the infinite-dimensional case requires Theorem \ref{Iloc} so as to ensure that $J^{-1}(\mu)/G_\mu$ is a submanifold.
  
\section{Symplectic forms on separable Hilbert spaces}
\label{sec:5}

Standard quantum mechanical theories assume Hilbert spaces to be separable, i.e. to admit a countable dense subset of vectors \cite{Neu_18}. So, we will skip some non-standard approaches \cite{TW_01,Neu_18} not assuming this condition. Let us now endow a separable Hilbert space with a strong symplectic form. This will provide a rigorous mathematical extension to the infinite-dimensional setting of the standard approach to finite-dimensional Hilbert spaces \cite{CCM_07,CL_11,CMMGM_16,CDIM_17}. We will indeed provide certain details that are frequently absent in the literature \cite{AT_99}.

A complex Hilbert space $(\mathcal{H}, \langle \cdot\vert \cdot \rangle)$ is separable if and only if it admits a countable orthonormal basis $\{\psi^j\}_{j\in\mathbb{N}}$ (see \cite{Rud_87}). This is quite a ubiquitous condition in quantum problems, but it may fail to hold, for instance, in scattering processes \cite{SN_20}. When $\mathcal{H}$ is separable, it admits an orthonormal basis \cite{Rud_87} that enables us to endow $\mathcal{H}$ with a real differentiable structure modelled on a Hilbert space. 

Choose real-valued coordinates $\{q^j, p_j\}_{j\in \mathbb{N}}$ on $\mathcal{H}$ given by
\begin{equation*}
    q^j(\psi):= \mathfrak{Re} \langle \psi^j \vert \psi \rangle\,, \qquad p_j(\psi) := \mathfrak{Im}\langle \psi^j \vert \psi\rangle\,,\qquad \forall j\in\mathbb{N}\,,\quad \forall \psi\in \mathcal{H}\,.
\end{equation*}
Let $\ell_2$ be the Hilbert space of complex square-summable series considered in the natural way as a real vector space. One can then construct a global chart $\varphi: \psi\in \mathcal{H}\mapsto (q^j(\psi),p_j(\psi))_{j\in \mathbb{N}}\in \ell^2$. In fact, every $\psi\in\mathcal{H}$ admits a unique family of complex constants $c_j\in\mathbb{C}$ such that $\psi = \sum_{j\in\mathbb{N}}c_j \psi^j$. It is worth recalling that since $\mathcal{H}$ is a Hilbert space $\lim_{j \rightarrow \infty} \sum_{j=1}^n c_j \psi^j = \psi$ and this holds independently of the order of the elements of the basis. This is not true for a Schauder basis of a Banach space. Consequently, $\Vert \psi \Vert_\mathcal{H} :=\sqrt{\langle \psi|\psi\rangle}= \sum_{j\in\mathbb{N}} \vert c_j \vert^2 < \infty $ and
\begin{equation}\label{Eq:Sym}
    \sum\limits_{j\in \mathbb{N}}\abs{c_j}^2 = \sum\limits_{j\in \mathbb{N}}\big({\mathfrak{Re}}(c_j)^2 + {\mathfrak{Im}}(c_j)^2\big) = \sum\limits_{j\in\mathbb{N}}\big(q^j(\psi)^2 + p_j(\psi)^2\big)<\infty\,.
\end{equation}
Thus, $\varphi(\psi)=(q^j(\psi),p_j(\psi))_{j\in \mathbb{N}}\in \ell^2$. Conversely, every sequence $(x_0^j,p_{0j})_{j\in \mathbb{N}}\in \ell ^2$ is the image of a unique element $\sum_{j\in\mathbb{N}} (x^j_0 + \ii p_{0j})\psi^j \in\mathcal{H}$. In view of \eqref{Eq:Sym}, one has that $\|\varphi(\psi)\|_{\ell^2}=\|\psi\|_\mathcal{H}$ and $\varphi$ is an isometry. Since $\varphi$ is an isometry and it has an inverse, its inverse is also an isometry and it is therefore continuous. Consequently, $\mathcal{H}$ and $\ell^2$ are homeomorphic. This gives rise to a real differentiable structure on $\mathcal{H}$. 

Recall that the tangent space to $\mathcal{H}$ at $\phi$ is the space of equivalence classes of curves $\gamma_{\psi}:\mathbb{R}\rightarrow \mathcal{H}$ such that $\gamma_{\psi}(0)=\phi$ and $\dd\gamma_\psi/\dd t(0)=\psi$. We will denote by $\dot\psi_\phi$ the equivalence class of the curve $\gamma_\psi$. Moreover, $\T_\phi\mathcal{H}=\{\dot\psi_\phi:\psi\in \mathcal{H}\}$ and there exists an isomorphism
$\lambda_{\phi}:\psi\in \mathcal{H}\mapsto \dot\psi_\phi\in  \T_{\phi}\mathcal{H}.$ Moreover, we represent $\lambda_\phi(\psi^j)$ and $\lambda_\phi ({\rm i}\psi^j)$ by $(\partial_{q^j})_{\phi}$ and $(\partial_{p_j})_{\phi}$, respectively. Hence, $\{(\partial_{q^j})_{\phi},(\partial_{p_j})_{\phi}\}_{j\in\mathbb{N}}$ becomes a basis of $\T_{\phi}\mathcal{H}$ and $\lambda_\phi(\psi) = \sum_{j\in \mathbb{N}}\left(q^j(\psi)(\partial_{q^j})_\phi+p_j(\psi)(\partial_{p_j})_\phi\right)$. It is worth noting that each $\T_\phi\mathcal{H}$ is then isomorphic to $\ell^2$.

The cotangent space of $\mathcal{H}$ at $\phi$, i.e. $\cT_{\phi}\mathcal{H} := (\T_{\phi}\mathcal{H})^\ast$, admits a basis $\{(\dd q^j)_{\phi}, (\dd p_j)_{\phi}\}_{j\in\mathbb{N}}$ dual to $\{(\partial_{q^j})_{\phi},(\partial_{p_j})_{\phi}\}_{j\in\mathbb{N}}$. Then,
\begin{equation*}
(\dd q^j)_\phi(\lambda_\phi(\psi))=q^j(\psi)=\mathfrak{Re}\langle \psi^j|\psi\rangle\,,\qquad (\dd p_j)_\psi(\lambda_\phi(\psi))=p_j(\psi)=\mathfrak{Im}\langle \psi^j\vert\psi\rangle\,.
\end{equation*}

At each $\phi$, one has then that $\T_\phi^*\mathcal{H}\simeq (\ell^2)^*\simeq \ell^2$. In other words, the elements of $\T^*_\phi\mathcal{H}$ are those of the form $\lambda_j(\d q^j)_\phi+\mu^j(\d p_j)$ with $(\lambda_j,\mu^j)_{j\in \mathbb{N}}$ forming a sequence in $\ell^2$.

Observe that $B:(\phi,\psi)\in \mathcal{H}\times \mathcal{H}\mapsto \mathfrak{Im}\langle \phi\vert\psi\rangle\in \mathbb{R}$ is a non-degenerate bilinear (over $\mathbb{R}$) form. Indeed, $B$ is well-defined because if ${
\bf c}_1{\bf c}_2\in \ell_2$, then $\sum_{n
\in \mathbb{N}} c_{1n}c_{2n}<\infty$ , and the map $B$ is bilinear and continuous. Since 
\begin{equation*}
B(\psi_1,\psi_2)=\mathfrak{Im}\langle \psi_1\vert\psi_2\rangle=\mathfrak{Im}\overline{\langle \psi_2\vert\psi_1\rangle}=-\mathfrak{Im}\langle \psi_2 \vert\psi_1\rangle=-B(\psi_2,\psi_1)\,,\qquad \forall \psi_1,\psi_2\in \mathcal{H}\,,
\end{equation*}
we obtain that $B$ is skew-symmetric. Then, $B$ allows us to endow $\mathcal{H}$ with a unique constant differential two-form $\omega$ such that $B=\lambda_\phi^*\omega_\phi$ on every $\phi\in \mathcal{H}$. By using the isomorphism  $\lambda_\phi$, one has that $\omega_{\phi}: \T_{\phi}\mathcal{H}\times \T_{\phi}\mathcal{H}\rightarrow \mathbb{R}$ takes the local form
\begin{equation}\label{Eq:Ex}
\omega_{\phi}(\dot \psi_1,\dot \psi_2) :=\mathfrak{Im}\langle \psi_1 \vert \psi_2 \rangle =\sum_{j\in \mathbb{N}} (q^j(\psi_1)p_j(\psi_2)-q^j(\psi_2)p_j(\psi_1))=\sum_{j\in \mathbb{N}}\left( \d q^j\wedge \dd p_j\right)(\dot\psi_1,\dot \psi_2)\,.
\end{equation}
Then, $\dd\omega=0$ and $\omega$ is a non-degenerate form.

To see that $\omega$ is a strong symplectic form, note that $\mathcal{H}$ admits a submanifold $\mathcal{Q}$ given by the zeroes of the functions $\{p^j\}_{j\in \mathbb{N}}$. In fact, the $\{q^j\}_{j\in \mathbb{N}}$ form a coordinate system on $\mathcal{Q}$ taking values in the real vector space of square-summable sequences $(q_j)_{j\in \mathbb{N}}$, making $\mathcal{Q}$ into a manifold modelled on a Hilbert space. Then, the variables $\{q^j,p_j\}_{j\in \mathbb{N}}$ can be understood as the adapted variables to $\cT\mathcal{Q}=\mathcal{H}$. In view of \eqref{Eq:Ex}, we get that $\omega$ is the canonical two-form of $\cT\mathcal{Q}$. Since $\cT\mathcal{Q}$ is modelled on a Hilbert space, which is always reflexive \cite{Rud_87}, Theorem \ref{RefCan} shows that $\omega$ is a strong symplectic form.

\section{Vector fields related to  operators}\label{sec:6}

In geometric finite-dimensional quantum mechanics, operators can easily be described via the so-called {\it linear vector fields}, which are smooth and even analytic \cite{CL_11,Mar_68}. Let us extend this idea to bounded, first, and next to general (possibly unbounded) operators on infinite-dimensional Hilbert spaces. This will involve addressing several issues concerning the lack of differentiability of the linear vector fields related to unbounded operators.

We showed in Section \ref{sec:5} that $\mathcal{H}$ is a smooth manifold modelled on a real Hilbert space through a global chart given by the coordinates $\{q^j,p_j\}_{j\in\mathbb{N}}$. This gives rise to a global coordinate system $\{q^j,p_j,\dot q^j,\dot p_j\}_{j\in\mathbb{N}}$ adapted to $\T\mathcal{H}$. This coordinate system gives rise to a diffeomorphism between $\T\mathcal{H}$ and $\mathcal{H}\oplus\mathcal{H}$. Then, every bounded operator $B:\mathcal{H}\rightarrow\mathcal{H}$ gives rise to a vector field
\begin{equation*}
X_{B}:\phi\in \mathcal{H}\mapsto (\phi,B\phi)\in \T_\phi\mathcal{H}\subset \T\mathcal{H}\,.
\end{equation*}
Since $B$ is bounded, it is continuous and $B:\mathcal{H}\rightarrow \mathcal{H}$ is also a smooth mapping. It follows immediately that $X_{B}$, which amounts to a bounded operator $\Id_\mathcal{H}\times B:\phi\in \mathcal{H}\mapsto (\phi,B\phi)\in \mathcal{H}\times \mathcal{H}$, where $\mathcal{H}\times \mathcal{H}$ is endowed with its natural Hilbert structure induced by the one on $\mathcal{H}$, is a smooth vector field.

In infinite-dimensional quantum mechanics, operators are standardly unbounded, they have a domain, and it is non-trivial to provide a procedure to ensure that they can be described somehow through linear vector fields on a domain that will be smooth enough for their differential geometric treatment. In particular, one knows that an  unbounded operator $B:D(B)\subset \mathcal{H}\rightarrow \mathcal{H}$   on a Hilbert space  $\mathcal{H}$ is related to a vector field with a domain \eqref{eq:VecOpe}. But it is important to discover what else can be said. 
Let us  detail this process by means of the theory of analytic vectors (see \cite{FSSS_72,Hal_13,Nel_59,Sch_12} for details).

\begin{definition}
\label{def:analyticvector}
	An {\it analytic vector} for an operator $B: D(B)\subset \mathcal{H} \rightarrow \mathcal{H}$ is an element $\phi\in \mathcal{H}$ such that $\phi\in D^\infty(B):=\bigcap_{n\in \bar{\mathbb{N}}}D(B^n)$ and $\| B^n \phi \| \leq C^n_\phi n!$ for some $C_x>0$ and every $n\in \mathbb{N}$\footnote{Note that there is a typo in \cite[Definition 7.1]{Sch_12}, which is used in this paper after the correction regarding $n$  values.}. The elements $\phi\in D^\infty(B)$ are called \textit{smooth} vectors \cite{Sim_72}.

\end{definition}

Both smooth and analytic vectors for a given operator $B$ are vector spaces. The fact that $D^\infty(B)$ is a vector space is clear from the linearity of $B$. Moreover if $\phi, \psi\in D^a(B)$, then simple calculation
$$
\| B^n (\phi + \psi) \| \leq \| B^n \phi \| + \| B^n \psi \| \leq C^n_\phi n! + C^n_\psi n! \leq \left( C_\phi + C_\psi \right)^n n! = C^n_{\phi+\psi} n!\,,
$$
shows that $D^a(B)$ is a vector space.

 We write $D^a(B)$ for the space of analytic vectors of $B$. The term  `smooth' vector comes from the fact that the mapping $\gamma:t\in\mathbb{R}\mapsto \exp(tB)\psi \in \mathcal{H}$ is smooth if and only if $\psi\in D^\infty(B)$ (see \cite[p. 145]{Sch_12}). Obviously, $D^a(B)\subset D^\infty(B)$.

The existence of analytic vectors for operators can be illustrated by \textit{Nelson's analytic vectors theorem}, which states that a  symmetric operator $A$ on a Hilbert space $\mathcal{H}$ is essentially self-adjoint if and only if $A$ admits a dense domain $D^a(A)$ of analytic vectors \cite{Hua_17,Nel_59}. We hereafter assume $B$ to be skew-self adjoint on some domain, if not otherwise stated. Analytic vectors allow for many other interesting simplifications in the theory of skew-adjoint operators. They are also useful in the differential geometric study of quantum mechanical problems. Let us also consider the following interesting result.

\begin{proposition} {\bf (\cite[Proposition 7.8]{Sch_12})}
    Let $A$ be a skew-self-adjoint operator and consider an analytic vector $\psi\in D^a(A)$. Let $C_{\psi}\in\mathbb{C}$ be such that $\|A^n\psi\|\leq |C_{\psi}|^nn!$ for every $n\in \mathbb{\mathbb{N}}$. If $z\in \mathbb{C}$ and $|z|<C_{\psi}^{-1}$, then $\psi\in D(e^{z {A}})$ and $e^{zA}\psi=\lim_{n\rightarrow +\infty}\sum_{k=0}^n\frac{z^k}{k!}A^k\psi. $
\end{proposition}

Since $e^{tA}$ is a continuous operator defined on $\mathcal{H}$, it stems from the definitions of smooth and analytic vectors for $A$ that $\exp(tA)D^\infty(A)\subset D^\infty(A)$ for every $t\in \mathbb{R}$ (cf. \cite[p. 149]{Sch_12}). In fact, if $\psi\in D^\infty(A)$, then $s\mapsto \exp(sA)\psi$ is smooth and $s\mapsto \exp((s-t)A)\exp(tA)\psi$ is again smooth. 



Assume that $V$ is a finite-dimensional Lie algebra of self-adjoint operators on $\mathcal{H}$. One may wonder when the exponential of the elements of $V$, which are continuous automorphisms on $\mathcal{H}$ by the Stone--von Neumann Theorem \cite{Hal_13,Sto_32}, can be understood  as  automorphisms on $\mathcal{H}$ related to a continuous Lie group representation on $\mathcal{H}$. The following definition, based on the one given in \cite{Fab_15} establishes the family of Lie algebras $V$ admitting such a property and defines rigorously what we meant.

\begin{definition} A finite-dimensional Lie algebra $V$ of operators with a common domain $\mathcal{D}$ within a Hilbert space $\mathcal{H}$ is {\it integrable} if there exists an injective Lie algebra representation $\rho:\mathfrak{g}\rightarrow V$ and a continuous Lie group action $\Phi:G\times \mathcal{H}\rightarrow \mathcal{H}$ (relative to the natural topologies on $G\times \mathcal{H}$ and $\mathcal{H}$), where $G$ is the connected and simply connected Lie group of $\mathfrak{g}$, such that
	\begin{equation*}
	\frac{\d}{\d t}\bigg|_{t=0}\Phi(\exp(tv),\psi) = \rho(v)\psi\,,\qquad \forall \psi\in \mathcal{D}\,,\qquad \forall v\in \mathfrak{g}\,.
	\end{equation*}
	We call $\rho$ the {\it infinitesimal Lie group action} associated with $\Phi$.
\end{definition}

Geometrically, the above definition tells us that the fundamental vector fields of the Lie group action $\Phi$ of $G$ on $\mathcal{D}$ coincide with the restriction of the elements of $V$ in $\mathcal{D}$. Moreover, the fundamental vector fields of this action satisfy that
\begin{equation}\label{eq:Comm}[X_{\rho(v_1)},X_{\rho(v_2)}]\psi=\frac{1}{2}\frac{\partial^2}{\partial t^2}\bigg|_{t=0}\exp(-t\rho(v_2))\exp(-t\rho(v_1))\exp(t\rho(v_2))\exp(t\rho(v_1))\psi\,, \qquad \forall \psi\in D\,.
\end{equation}
where $D$ is an invariant common domain of smooth functions where the operators exponentials of the elements of $V$ are well defined and smooth enough. In other words, despite the lack of differentiability of the vector fields related to operators, the geometric expression of the Lie derivative of vector fields still holds.

\begin{example} Consider the operators $\ii \hat p$, $\ii \hat x$ and $\ii \widehat{{\rm Id}}$ on $L^2(\mathbb{R})$ spanning a three-dimensional Lie algebra of skew-self-adjoint operators. Consider the subspace of analytic functions of the Schwarz space $\mathcal{S}(\mathbb{R})$. This Lie algebra of operators admit a dense subspace of analytic vectors given by the functions $H_n(x)e^{-x^2/2}$ and it can be integrated. Moreover, these functions belong to $\mathcal{S}(\mathbb{R})$, which is invariant relative to the action of $\exp(\lambda \ii \widehat p),\exp(\lambda \ii \widehat x)$ and $\exp(\lambda \ii \widehat {\rm Id})$ for $\lambda \in \mathbb{R}$. For instance,
$$
(\exp(\ii t\widehat{p})\psi)(x)=\psi(x-t),\qquad (\exp(\ii t\widehat{x})\psi)(x)=e^{\ii x}\psi(x),\qquad \forall x\in \mathbb{R}. 
$$
Hence, expression \eqref{eq:Comm} gives
$$
([X_{\ii\widehat{p}},X_{\ii\widehat{x}}]\psi)(x)=\frac 12\frac{\partial^2}{\partial t^2}\bigg|_{t=0}\exp(-\ii t \widehat{x})\exp(-\ii t  \widehat{p})e^{\ii t x}\psi(x-t)=\frac 12\frac{\partial^2}{\partial t^2}\bigg|_{t=0}\exp(\ii t^2)\psi(x)=\ii\psi(x)=X_{\ii\widehat{\rm Id}}.
$$
\end{example}

It is a consequence of an argument given by Nelson \cite{Nel_59} and Goodman \cite{Goo_69} that if a finite-dimensional real Lie algebra of skew-symmetric operators is integrable, then the next criterium holds (cf. \cite{FSSS_72}). 

\begin{theorem} {\bf (FS$^3$ criterium \cite{FSSS_72})} 
Every finite-dimensional Lie algebra of skew-symmetric operators on a Hilbert space $\mathcal{H}$ is integrable if it admits a basis with a common dense invariant subspace of analytic vectors.
\end{theorem}

There are other criteria of integrability of finite-dimensional Lie algebras of skew-symmetric operators on Hilbert spaces \cite{Nel_59,Sim_72}, but they are, in general, more complicated to use. On the other hand, the FS$^3$ criterium can be relatively easily proved to hold in many physically relevant problems appearing in physics (cf. \cite{CL_11, CLR_09, CR_03,CR_05}). In particular, if $A_1,\ldots,A_n$ are observables admitting a common basis of eigenvectors $\{e_n\}_{n\in \mathbb{N}}$, then $\ii A_1,\ldots,\ii A_n$ admit a basis of eigenvectors $\{e_n\}_{n\in \mathbb{N}}$ and all their elements are analytic relative to $\ii A_1,\ldots, \ii A_n$. Moreover, the direct sum of subspaces $\bigoplus_{n\in \mathbb{N}}\langle e_n\rangle\simeq \mathbb{R}^\mathbb{N}$ is a normed not complete space (in general) of analytic vectors relative to $\ii A_1,\ldots,\ii A_n$ that is invariant relative to the operators $\exp(\ii \lambda_\alpha A_\alpha)$ with $\lambda_\alpha\in \mathbb{R}$ and $\alpha=1,\ldots,r$.

Assume then that the skew-symmetric operators of the Lie algebra $V$ have a common invariant domain $D_V$ of analytic vectors on $\mathcal{H}$. The related continuous unitary action $\Phi:G\times\mathcal{H}\rightarrow\mathcal{H}$ satisfies that $\Phi_\psi:g\in G\mapsto \Phi(g,\psi)\in \mathcal{H}$ is an analytical function on the space of analytic elements $D_V$. If $D_V^\infty$ is a common invariant domain of smooth vectors for the elements of $V$, then $\Phi_\psi$ is smooth on $D_V^\infty$ (cf. \cite{Har_53})

Let us now prove the following useful result. 

\begin{theorem}
    Every observable $B$ on a Hilbert space $\mathcal{H}$ gives rise to an integrable vector field  on a domain $X_{B}:\psi\in D^{\infty}(B)\subset  \mathcal{H}\mapsto (\psi,B\psi)\in \mathcal{H}\oplus\mathcal{H}\simeq \T\mathcal{H}$ and $X_{ B}(\psi)\in \T_{\psi}D^{\infty}(B)$.
\end{theorem}
\begin{proof}
    We know that $\T\mathcal{H}$ is diffeomorphic to $\mathcal{H}\oplus\mathcal{H}$. Hence, $B$ can be associated with the vector field on a domain given by $X_{B}$. Note that the direct sum of the subspaces generated by the eigenvalues of $B$ gives a dense subspace $E_{B}\simeq \mathbb{R}^{\mathbb{N}}$ of smooth vectors for $B$ in $\mathcal{H}$. Hence, $D^\infty(B)$ is not empty and becomes dense in $\mathcal{H}$. Moreover, $D^a(B)\supset E_{B}$ is also dense in $\mathcal{H}$ and the Nelson's analytic vectors theorem shows $B$ is integrable. The norm of $\mathcal{H}$ can be restricted to $D^{\infty}(B)$, which becomes a normed space. The space $E_{B}$ is dense in $D^{\infty}(B)$ and we obtain a basis for $D^{\infty}(B)$. The maps of the global atlas on $\mathcal{H}$ can be restricted to $D^{\infty}(B)$, which makes $D^\infty({B})$ into a manifold modelled on a normed space. The inclusion $\jmath:D^{\infty}(B)\hookrightarrow \mathcal{H}$ is smooth relative to the previous differential structures because every smooth function on $\mathcal{H}$ remains smooth when restricted to $D^{\infty}(B)$.  In fact, the variations of a smooth function on $\mathcal{H}$ remain smooth the directions tangent to $D^{\infty}(B)$ at points of $D^{\infty}(B)$.

    On an element $\psi\in D^\infty(B)$, one has that $F_\psi$ is smooth. In view of  this, one can see that $X_{B}$ is a vector field on the domain $D^{\infty}(B)$ and $B(D^\infty(B))\subset D^\infty(B)$ (see \cite{Har_53} for further details). Hence, $X_{B}$ is an integrable vector field on the domain $D^\infty(B)$. Since $B(D^\infty(B))\subset D^\infty(B)$, one has that $X_{B}$ can be considered also as a vector field on a domain $D^\infty(B)$ taking values in $\T D^\infty(B)$.
\end{proof}

 It is worth studying non-symmetric operators for which a vector field approach is possible. Let us give an example.

 \begin{example} Consider the space of square integrable functions on $[-l,l]\subset \mathbb{R}$ for $l>0$, which is a Hilbert space, and the operator $-\ii\widehat{p}^2:D(-\ii\widehat{p}^2)\subset \mathcal{H}\rightarrow \mathcal{H}$. In particular, consider $D(-\ii\widehat p^2)$ to be the space spanned by finite linear combinations of functions 
 $$
 \frac{1}{\sqrt{l}}\sin\left(\frac{\pi (k+1) x}l\right),\qquad  \frac{1}{\sqrt{l}}\cos\left(\frac{\pi k x}l\right),\qquad k=0,1,2,3,\ldots  
 $$
 Hence, $D(-\ii\widehat{p}^2)$ becomes a normed space relative to the norm induced by $L^2([-l,l])$  that is isomorphic to $\mathbb{R}^\mathbb{N}$  and dense in $L^2([-l,l])$. Then, one obtains an operator $X_{-\ii\widehat{p}^2}:\psi\in D(-\ii\widehat{p}^2)\subset \mathcal{H}\mapsto (\psi,-\ii\widehat{p}^2\psi)\in \T \mathcal{H}$. It is worth stressing that $X_{-\ii  \widehat{p}^2}$ is not smooth and $-\ii\widehat{p}^2$ is not symmetric. But one may define a flow
$$
\varphi:(t,\psi)\in \mathbb{R}\times D(\widehat{p}^2)\mapsto \exp(-\ii t \widehat{p}^2)\psi\in \mathcal{H}.
$$
Note that this map is well defined as
$$
\exp(-\ii t\widehat{p}^2)\sin(\pi k x/L)=\exp(\ii t\pi^2k^2/L^2)\sin(\pi k x/L),
$$
$$
\exp(-\ii t\widehat{p}^2)\cos(\pi k x/L)=\exp(\ii t\pi^2k^2/L^2)\cos(\pi k x/L).
$$
Moreover, $D(-\ii \widehat{p}^2)$ consists of analytic vectors for $-\ii\widehat{p}^2$ and $\exp(-\ii t\widehat{p}^2)$ is continuous on $D(-\ii t\widehat{p}^2)$.
 \end{example}

\section{Hamiltonian functions induced by self-adjoint operators}\label{sec:7}

Previous section showed that vector fields generated by skew-self adjoint operators may be well-defined on an appropriate domain. We will show that these vector fields can be considered as Hamiltonian vector fields in a domain by using functions on domains. For instance, we will have to deal with functions on Hilbert spaces of the type 
\begin{equation*}
f_H:\psi\in D(H)\subset \mathcal{H}\mapsto \langle \psi \vert H\psi\rangle\in\mathbb{R}\,,
\end{equation*}
where $H$ is a skew-symmetric operator. The function $f_H$ is called the {\it average value} of the operator $H$ \cite{Hal_13} and it is only defined on the domain of $H$, which is dense in $\mathcal{H}$. 

In general, the whole symplectic formalism can be modified to deal with quantum mechanical systems by considering that structures are defined on a dense subset of $\mathcal{H}$ having the structure of a normed space and using the fact that the operators of the theory are self- or skew-self-adjoint. Most structures to be used are not smooth, but they are differentiable enough to apply differential geometric techniques.

Let us start by defining how to write in terms of quantum operators the differentials of the real and imaginary parts of certain average values of operators. This is done in the following lemma.

\begin{lemma}\label{Le:AF} If $f_A$ is the average function of a symmetric operator $A:D(A)\subset \mathcal{H}\rightarrow \mathcal{H}$, then 
	\begin{equation*}
	\inn{\dot{\phi}}(\dd f_A)_\psi=2\mathfrak{Re} \langle \phi \vert A \psi\rangle\,,\qquad \forall \dot{\phi}\in \T_\psi D(A)\,.
	\end{equation*}
Therefore, $(\dd f_A)_{\psi}$ can be extended to a unique element of $\cT_\psi\mathcal{H}$.
\end{lemma}

\begin{proof}
By definition of the contraction of the differential of a function with a vector field
\begin{equation*}
\iota_{\dot\phi}(\dd f_A)_\psi= \lim_{t\rightarrow 0 } \frac{f_A(\psi +t\phi) - f_A(\psi)}{t} = \langle \psi\vert A \phi\rangle + \langle \phi\vert A\psi\rangle=2\mathfrak{Re} \langle \phi \vert A \psi\rangle\,,
\end{equation*}
for all $\phi,\psi\in D(A)$. Now, observe that $(\dd f_A)_{\psi}(\dot{\phi})=2\mathfrak{Re}\langle A\psi \vert \phi \rangle \leq 2 \Vert A\psi\Vert\Vert \phi \Vert$, so that $(\dd f_H)_{\psi}$ is bounded on $D(A)$. Since $D(A)$ is a dense subspace of  $\mathcal{H}$ and $(\dd f_A)_\psi(\dot{\phi})$ is bounded on $D(A)$, then, in view of the Hahn--Banach theorem, $(\dd f_A)_{\psi}$ may be continuously extended to a linear continuous form on $\T_{\psi}\mathcal{H}$. Since $D(A)$ is dense in $\mathcal{H}$, this extension is unique.
\end{proof}

The above lemma shows that $f_A$ is Fr\'{e}chet differentiable at every point of $D(A)$ and it allows us to introduce the following definition.

\begin{definition} Let $f_H$ be the average value of a symmetric operator $H:D(H)\subset \mathcal{H}\rightarrow \mathcal{H}$. Its {\it differential} is the mapping $\dd f_H:\psi\in D(H)\subset \mathcal{H}\mapsto (\dd f_H)_\psi\in \cT_\psi\mathcal{H}\subset \cT\mathcal{H}$, where $(\dd f_H)_\psi(\dot\phi)=2\mathfrak{Re}\langle\phi|H\psi\rangle$.
\end{definition}

Definitions from symplectic geometry can be adapted to the context of quantum mechanics by assuming that they must be restricted appropriately to domains. 

\begin{definition}
    We say that a vector field on a domain $X:D(H)\subset \mathcal{H}\rightarrow \T\mathcal{H}$ is {\it Hamiltonian} relative to a symplectic manifold $(\mathcal{H},\omega)$ if $\inn{X}\omega=\dd h$, where the right-hand side is understood as the G\^ateaux differential of a Fréchet differentiable function $h:D(H)\subset\mathcal{H}\rightarrow \mathbb{R}$, which is defined at points of $D(H)$.
\end{definition}
It is worth noting that the mapping $\dd h_\psi$ does not need to exists for $\psi\notin D(H)$, while the contraction is in points of the domain of $X$.
\begin{proposition}\label{Pro:VH}
    The vector field with a domain $X_{-\ii H}$ on $\mathcal{H}$ induced by a symmetric operator $H:D(H)\subset \mathcal{H}\rightarrow \mathcal{H}$, where $D(H)$ is a domain of $H$, is Hamiltonian relative to the Hamiltonian function $f_H:\psi\in D(H)\mapsto\frac12\langle\psi \vert H\psi \rangle\in \mathbb{R}$. 
\end{proposition}
\begin{proof} Using the expression for the canonical symplectic form on $\mathcal{H}$, we obtain, in view of Lemma \ref{Le:AF}, that 
	$$
	(\inn{X_{-\ii H}}\omega)_\psi(\dot \phi)=\mathfrak{Im}\langle -\ii H  \psi\vert\phi\rangle=\mathfrak{Re}\langle H \psi\vert\phi\rangle=\frac12\inn{\dot\phi}\dd \langle \psi|H\psi\rangle\,,\qquad \forall \dot\phi\in \T_\psi H\,,\quad\forall \psi\in \mathcal{D}(\mathcal{H})\,.
	$$
\end{proof}
\section{On \texorpdfstring{$t$}--dependent Hamiltonian operators}\label{sec:8}

To make a differential geometric approach to $t$-dependent Schr\"odinger equations possible, this section presents a specific type of $t$-dependent Schr\"odinger equations to be studied. In particular,  we are concerned with  the ones of the form
\begin{equation}
\label{eq:Schr}
\frac{\partial \psi}{\partial t} = -\ii H(t)\psi\,,\qquad \psi\in \mathcal{H}\,,
\end{equation}
where the $t$-dependent $H(t)$ are determined by the following proposition. 

\begin{definition}\label{def:TDH}
    A {\it $t$-dependent Hamiltonian operator} is a time-parametrised family of self-adjoint operators $H(t)$ on a Hilbert space $\mathcal{H}$ such that all the $H(t)$ admit a common invariant domain $\mathcal{D}$ of analytic vectors in $\mathcal{H}$.
\end{definition}

Note that we do not impose each $H(t)$ to be self-adjoint on $\mathcal{D}$. Since $\mathcal{D}$ will be a subset of the domains $D(H(t))$, with $t\in \mathbb{R}$, where each $H(t)$, for a fixed $t$, is self-adjoint for a certain domain $D(H(t))$, the operators $H(t)$ are only symmetric on $\mathcal{D}$. Definition \ref{def:TDH} is satisfied by relevant physical systems. The following example shows one of them.

\begin{example}\label{Ex:MP}The quantum operators 
$$
	H_1:= \hat{p}^2\,,\qquad H_2:=\hat{x}^2\,,\qquad H_3:= (\hat{x}\hat{p}+\hat{p}\hat{x})\,,\qquad H_4:= \hat{p}\,,\qquad H_5:= \hat{x}\,,\qquad H_6:= \Id
$$
are self-adjoint operators on $L^2(\mathbb{R})$ with respect to appropriate domains (cf. \cite{Gos_11,Hal_13,Sch_12}). The operators $H_1,\ldots,H_6$ span a Lie algebra of skew-symmetric operators on a common domain $\mathcal{D}$ of functions of the form $e^{-x^2/2}P(x)$, where $P(x)$ is any polynomial (cf. \cite{Sch_12}). The space $\mathcal{D}$ admits an algebraic basis $e^{-x^2/2}H^p_n(x)$ with $n=0,1,2,\ldots$, where the $H^p_n(x)$ are the so-called {\it Hermite probabilistic polynomials.} The functions $e^{-x^2/2}H_n^p(x)$ constitute a basis of $L^2(\mathbb{R})$ and, then, $\mathcal{D}$ is dense in $L^2(\mathbb{R})$ (see \cite{Sze_39}). The operators $H_1,\ldots,H_6$ are symmetric on $\mathcal{D}$. 
Moreover, $\mathcal{D}$ is also invariant relative to the action of  $H_1,\ldots,H_6$. Hence, every $t$-dependent Hamiltonian operator on $L^2(\mathbb{R})$ of the form
	\begin{equation}\label{Eq:Osc}
	H(t):=\sum_{\alpha=1}^6b_\alpha(t)H_\alpha\,,
	\end{equation}
where the functions $b_\alpha(t)$ are any real $t$-dependent functions, is a $t$-dependent Hamiltonian operator with a domain $\mathcal{D}$. 
 
 The $t$-dependent Hamiltonian \eqref{Eq:Osc} embraces $t$-dependent harmonic oscillators and other related systems \cite{CL_11}. Note that $V=\langle  \ii H_1,\ldots,  \ii H_6\rangle$ is a Lie algebra of skew-symmetric operators on $\mathcal{D}$.  Moreover, if one restricts to the functions $\psi_m(x):=x^me^{-x^2/2}$, where $m=0,1,2,3,\ldots,$, one finds after a simple calculation that these elements are analytic for  $H_4,H_5,H_6$. The constant bound for a vector $\psi_m$ according to the Definition \ref{def:analyticvector} of analytic vector could be, for instance $C_{\psi_m}= 2^{m+1}m!$, for all the operators $H_4,H_5,H_6$. Then, $\mathcal{D}$ consists of analytic vectors for $H_4,H_5,H_6$. Meanwhile, it can be proved that $H^p_m(x)e^{-x^2/2}$ are also analytic for $H_1,H_2,H_3$. Hence, the elements of $\mathcal{D}$ are analytic for $H_1,\ldots,H_6$. Therefore, $H_1,\ldots,H_6$ can be integrated giving rise to the so-called {\it metaplectic representation} of the {\it metaplectic group}.

 It is worth noting that $\mathcal{D}$ is not invariant relative to the operators of the form $\exp(i\lambda H_\alpha)$ for any real $\lambda$. For instance, $\exp(\lambda H_4)(x^me^{-x^2/2})=(x+\lambda)^me^{-(x+\lambda)^2/2}$ and the latter does not belong to $\mathcal{D}$. 	
 \demo
\end{example}
\begin{example}  Consider the space 
$$
\mathcal{D}=\{P(x,y,z)e^{-x^2/2-y^2/2-z^2/2}\in L^2(\mathbb{R}^3)\mid P(x,y,z)\ \text{is a polynomial}\}\,,
$$
and the operators
$$
\widehat{L}_x=yp_z-zp_y\,,\qquad
\widehat{L}_y=zp_x-xp_z\,,\qquad
\widehat{L}_z=xp_y-yp_x\,.
$$
Using an appropriate basis of polynomials in $x,y,z$, one can prove that $\mathcal{D}$ is a space of analytic vectors for the above operators.
\end{example}
$t$-dependent Schr\"odinger equations on $\mathcal{H}$ describe the dynamics of quantum systems and the elements of $\mathcal{H}$ represent its quantum states. Under the conditions given in Definition \ref{def:TDH} on the associated $H(t)$, we can prove that the $t$-dependent Schr\"odinger equation can be considered as a Hamiltonian system.

Using the diffeomorphism $\T\mathcal{H}=\mathcal{H}\oplus\mathcal{H}$, we can understand the solutions to \eqref{eq:Schr} as curves in $\mathcal{H}$ whose tangent vectors at $\psi(t)$ are given by $X(t,\psi(t))=(\psi(t),-\ii H(t)\psi(t))$. Note that this equality only makes sense in the space $\mathcal{D}$ which is included in the space of smooth vectors for all the operators $H_1,\ldots,H_r$. Recall also that if $\psi\in \mathcal{H}$ is smooth, then $\exp(tH_i)\psi$ is also smooth. This gives rise to a $t$-dependent vector field with a domain $X: (t,\psi)\in \mathbb{R}\times D\mapsto X(t,\psi)\in \T\mathcal{H}$. The $t$-dependent vector field on a domain, $X$, is well-defined since it is a linear combination with $t$-dependent functions of the vector fields related to the skew-symmetric operators related to $-\ii H_\alpha$, with $\alpha=1,\ldots,r$. 

Finally, Proposition \ref{Pro:VH} ensures that the $t$-dependent vector field $X$ is, for every $t\in \mathbb{R}$, a Hamiltonian vector field with Hamiltonian function $\frac12\langle\psi|H(t)\psi\rangle$. As a consequence, we obtain the following theorem.

\begin{theorem}\label{Pro:ST}
    A $t$-dependent Schr\"odinger equation related to a $t$-dependent Hamiltonian $H(t)$ is the differential equation for the integral curves of the $t$-dependent vector field on a domain associated with the $t$-dependent Hamiltonian system $(\mathcal{H},\omega,\frac12\langle\psi|H(t)\psi\rangle)$.
\end{theorem}

The domain $\mathcal{D}$ of analytic vectors is not invariant relative to the exponential of the operators $H_1,\ldots,H_r$. This happens, for instance, in Example \ref{Ex:MP}. There are other many case where this can be assumed.
\begin{definition}
	A {\it quantum $t$-dependent Hamiltonian system} is a Hamiltonian system of the form $(\mathcal{H},\omega,\frac12\langle \psi| H(t)\psi\rangle)$, where $H(t)$ is a $t$-dependent Hamiltonian operator on $\mathcal{H}$.
\end{definition}
Hence, the above states that the $t$-dependent Schr\"odinger equations defined by a $t$-dependent Hamiltonian are the Hamilton equations of the $t$-dependent Hamiltonian.

\section{Reduction process for {\it t}-dependent Schr\"odinger equations}\label{sec:9}

This section describes the projection of a $t$-dependent Schr\"odinger equation determined by a $t$-dependent Hamiltonian operator, which satisfies appropriate technical conditions, onto a projective space via a Marsden--Weinstein reduction. 

Recall that the projective space, ${\mathcal{PH}}$, of a Hilbert space $\mathcal{H}$ is the space of equivalence classes of proportional non-zero elements of $\mathcal{H}$. Every equivalence class in ${\mathcal{PH}}$ is called a {\it ray} of $\mathcal{H}/\{0\}$. Physically, the projective space ${\mathcal{PH}}$ is what is really relevant, as different elements of $\mathcal{H}$ belonging to the same ray have the same physical meaning \cite{SN_20}. 
 
More particularly, let us apply the Marsden--Weinstein reduction to a $t$-dependent Schr\"odinger equation determined by a $t$-dependent Hamiltonian operator $H(t)$ on a Hilbert space $\mathcal{H}$ to describe its projection onto the projective space $\mathcal{PH}$ as a $t$-dependent Hamiltonian system relative to a reduced symplectic form. To accomplish the process for $t$-dependent Schr\"odinger equations related to unbounded operators, we  first define what an invariant function is. As in previous sections, to define it, one just considers an appropriate restriction of the infinite-dimensional Hilbert space to an appropriate normed domain, as it happens in quantum mechanics. With this respect, it is convenient to recall what we mean by a symmetry of a function relative to a Lie group action.

\begin{definition}\label{Def:GInv}
    Let $\Phi: G \times \mathcal{H}\rightarrow \mathcal{H}$ be a Lie group action and let $D$ be a dense complex subspace of $\mathcal{H}$. A $t$-dependent function $f: (t,\psi)\in \mathbb{R}\times D\mapsto f_t(\psi):=f(t,\psi)\in \mathbb{R}$ is said to be \textit{invariant} under $\Phi$ (or $G$-\textit{invariant} when $G$ is known from context) if $f_t\circ \Phi_g = f_t$ for all  $g\in G$ and every $t\in\mathbb{R}$.
\end{definition}

Implicitly, Definition \ref{Def:GInv}  implies that $\Phi_g(D)\subset D$ for every $g\in G$. Otherwise, if $\Phi_g(\psi)\notin D$ for certain $g\in G$ and $\psi\in D$, it would follow that $f(\Phi_g(\psi)) \neq f(\psi)$ since the left-hand side of the equality would not be defined. Naturally, one says that a $t$-independent function $f:D\rightarrow \mathbb{R}$ is $G$-invariant relative to a Lie group action $G\times \mathcal{H}\rightarrow \mathcal{H}$ if $\jmath^*f$, where  $\jmath:(t,\psi) \in \mathbb{R}\times D\mapsto \psi\in D$, is $G$-invariant in the sense given in Definition \ref{Def:GInv}.

\begin{example}
    Consider the multiplicative group  $U(1)=\{e^{\ii\theta}\! \enskip\! \vert\!\! \enskip \theta\in[0,2\pi)  \}$ and the Lie group action $\Phi: (e^{\ii \theta},\psi) \in U(1) \times \mathcal{H} \mapsto e^{\ii \theta}\psi\in \mathcal{H}$. Let us prove that if $H$ is a symmetric operator with domain $D(H)$, then its average function $f_H$ is $U(1)$-invariant. In fact, the domain $D(H)$ is a $\mathbb{C}$-linear subspace of $\mathcal{H}$. Moreover,
    \begin{equation*}
    	f_H \circ \Phi_{e^{{\rm i}\theta}} (\psi) = \langle e^{\ii\theta}\psi \vert H  e^{\ii\theta}\psi \rangle = e^{-\ii\theta}e^{\ii\theta} \langle \psi \vert H \psi \rangle = f_H(\psi)\,,\qquad \forall \psi\in D(H)\,.
	\end{equation*}
	Therefore, $f_H$ is invariant relative to  $\Phi_{e^{{\rm i}\theta}}$ for every $e^{{\rm i}\theta} \in U(1)$. Recall that since $H$ is symmetric, $f_H$ is a real function. Consequently, $f_H$ is a $U(1)$-invariant function. It is worth noting that, if $\psi\in D(H)$, then $e^{\ii \theta}\psi\in D(H)$ for every $\theta\in \mathbb{R}$ and $\Phi_g(\psi)\in D(H)$ for every $g = e^{{\rm i}\theta}\in U(1)$. Hence, $D(H)$ is invariant under $\Phi$. \demo
\end{example}

Let us now reduce a $t$-dependent Schr\"odinger equation \eqref{eq:Schr} associated with a $t$-dependent Hamiltonian operator $H(t)$ on a separable Hilbert space $\mathcal{H}$ onto its projective Hilbert space, $\mathcal{P}\mathcal{H}$, via a Marsden--Weinstein symplectic reduction. 

Consider global real-valued coordinates $\{q^j,p_j\}_{j\in\mathbb{N}}$ on $\mathcal{H}$ as mentioned in Section \ref{sec:5}. The initial data for the reduction is given by  
\begin{equation*}
    \omega:= \sum\limits_{j\in \mathbb{N}}\dd q^j \wedge \dd p_j\,, \qquad H(t):=\sum\limits_{\alpha=1}^r b_{\alpha}(t)H_{\alpha}\,,
\end{equation*}
where the $b_1(t),\ldots,b_r(t)$ are arbitrary $t$-dependent real-valued functions and $H_{\alpha}$, with $\alpha = 1,\dotsc,r$, are self-adjoint operators such that all elements of $\langle H_1, \ldots ,H_r \rangle$ share a common domain of analytic vectors $D\subset\mathcal{H}$. We stress that we do not require  $H_1,\ldots,H_r$ to be self-adjoint on $D$. Moreover, note that $\{\d q^j,\d p_j\}_{\mathbb{N}}$ is a basis of $\T^*\mathcal{H}$ and every element of each $\T^*_{(q,p)}\mathcal{H}$ is the limit of a sum of finite linear combinations of elements of $\{\d q^j,\d p_j\}_{\mathbb{N}}$. 

The $t$-dependent Schr\"{o}dinger equation \eqref{eq:Schr} induced by $H(t)$ defines a $t$-dependent Hamiltonian vector field on $\mathcal{H}$ with a domain $D\subset \mathcal{H}$ of the form 

\begin{equation*}
X(t,\psi):=\sum_{\alpha=1}^r b_{\alpha}(t) X_{\alpha}(\psi)\,,\qquad \forall \psi \in D\,,\qquad \forall t\in \mathbb{R},
\end{equation*}
such that $X_{\alpha}(\psi):=-\ii H_{\alpha}\psi$ for all $\alpha=1,\dotsc,r$ and $\psi\in D$. We will hereafter consider the reduction of $\omega$ to the projective space, and then the reduction of the $t$-dependent Schr\"{o}dinger equation associated with $H(t)$.

\subsection{The Marsden--Weinstein reduction of the symplectic form on \texorpdfstring{$\mathcal{H}$}{}}

Consider the Lie group action $\Phi: U(1) \times \mathcal{H}\ni (e^{\ii\theta \Id_{\mathcal{H}}},\psi) \mapsto e^{\ii \theta}\psi\in\mathcal{H}$. This Lie group action is symplectic relative to the symplectic form \eqref{Eq:Ex}, as proved in Example \ref{Ex:expaction}.

To obtain the momentum map associated with $\Phi$, we must relate every fundamental vector field of $\Phi$ to a Hamiltonian function.
The Lie algebra $\mathfrak{u}(1)$ of $U(1)$ consists of pure imaginary numbers, that is $\mathfrak{u}(1)=\{a\ii \mid a\in\mathbb{R}\}$. The fundamental vector field of $\Phi$ generated by $\ii\in\mathfrak{u}(1)$ reads
\begin{equation*}
    X_{\ii}(\psi) = \frac{\dd}{\dd t}\bigg\vert_{t=0}\Phi(e^{\ii t},\psi) = \ii \psi\,, \qquad \forall\psi\in\mathcal{H}\,.
\end{equation*}
This is the vector field associated with the skew-self-adjoint and bounded operator $A=\ii \Id$. In view of these results and Proposition \ref{Pro:VH}, an appropriate momentum map $J$  corresponding to the action $\Phi$ is given by
\begin{equation*}
J:\mathcal{H}\ni\psi \longmapsto - \frac{1}{2} \langle \psi \vert\psi\rangle\in \mathfrak{u}^*(1)\simeq \mathbb{R}^*\,,
\end{equation*}
so that the values of $J$ at every $\psi\in\mathcal{H}$ are proportional to the square of the norm of $\psi$. Let us prove that $J$ is smooth. The directional derivative of $J$ in the direction ${\psi}$ at $\psi_0$, namely $\dot\psi_{\psi_0}J$, becomes 
\begin{equation*}
\dot\psi_{\psi_0}J=\lim_{t\rightarrow 0} \frac{J(\psi_0+t\psi)-J(\psi_0)}t=-\mathfrak{Re}\langle \psi_0\vert \psi\rangle\,, \qquad \forall\psi_0\in \mathcal{H}\,, \quad \forall\dot{\psi}_{\psi_0}\in \T_{\psi_0}\mathcal{H}\,.
\end{equation*}
This gives rise to a linear and bounded Gateaux differential of $J$  on $\mathcal{H}$ for each $\psi_0\in\mathcal{H}$. Hence, $J$ is differentiable (see \cite{Col_12,Mar_68,Sch_22} for details). Let us write $\nabla_{\dot\psi}J$ for the partial derivative of $J$ in the direction $\dot\psi$. The second-order directional derivative of $J$ along $\phi$ and $\psi$ is constant: 
\begin{align*}
    [\nabla_{\dot\phi}\nabla_{\dot\psi}J]_{\psi_0}&= \lim_{t\rightarrow 0} \frac{1}{t}\big((\nabla_{\dot{\psi}} J)(\psi_0+t\phi)-(\nabla_{\dot{\psi}}J)(\psi_0)\big)\\
    &= -\lim_{t\rightarrow 0} \frac{1}{t}\big(\mathfrak{Re} \langle \psi_0 + t\phi \vert\psi\rangle -\mathfrak{Re}\langle \psi_0 \vert \psi \rangle\big) = -\mathfrak{Re}\langle \phi \vert \psi\rangle\,, \qquad \forall\psi_0\in\mathcal{H}\,,\quad \forall \dot\psi, \dot\phi \in \T_{\psi_0}\mathcal{H}\,.
\end{align*}
Hence, all second-order directional derivatives are constant and continuous relative to $\dot\psi,\dot\phi$. Higher-order derivatives are zero. Thus, $J$ is smooth. 

Let us show that $\Phi$ is $\Ad^*$-equivariant. Since $U(1)$ is an Abelian Lie group, its adjoint action $\Ad$ is trivial, that is $\Ad_g=\Id_{\mathfrak{u}(1)}$ for every $g\in U(1)$. Then, the coadjoint action is also trivial. Moreover, $J$ is invariant with respect to the action of $U(1)$ on $\mathcal{H}$, since unitary transformations do not change the norm of elements of the Hilbert space $\mathcal{H}$. Therefore,
\begin{equation*}
    J(\psi)=J\circ \Phi_{g}(\psi)=  \Id_{\mathfrak{u}^*(1)}\circ J(\psi)=(\Ad^*)_{g} \circ J(\psi)\,,\qquad\forall \psi\in \mathcal{H}\,,\quad \forall g\in U(1)\,.
\end{equation*}
Thus, the momentum map $J$ is $\Ad^*$-equivariant.

Let us prove that every $\mu<0$ is a regular value of $J$. Since $J$ is smooth, it is enough to prove that $\T_{\psi}J: \T_{\psi}\mathcal{H} \rightarrow \T_{J(\psi)}\mathbb{R}^*\simeq \mathbb{R}^*$ is surjective with a split kernel for all non-zero elements $\psi\in\mathcal{H}$ (see \cite{AMR_88}). The tangent map of $J$ at $\psi\in\mathcal{H}$ is $\T_{\psi}J(\dot\phi)={\dot\phi}_\psi J=-\mathfrak{Re}\langle \psi\vert \phi\rangle$.
Note that $\T_{\psi}J\neq 0$ if and only if $\psi\neq 0$. Therefore, $\Ima \T_{\psi}J=\T_{J(\psi)}\mathbb{R}^*$ for $\psi\neq 0$ and thus $\T_{\psi}J$ is surjective if and only if $\psi\in \mathcal{H}\backslash\{0\}$. Since $J$ is smooth, $\T_{\psi}J$ is a continuous map, $\ker \T_{\psi}J$ is a closed subspace of $\T_{\psi}\mathcal{H}$, which becomes a Hilbert space relative to $\mathfrak{Re}\langle \cdot\vert\cdot\rangle$. Moreover, $\T_{\psi}J$ has a split kernel, since $\ker\T_{\psi}J\oplus (\ker\T_{\psi}J)^{\perp} = \T_{\psi}\mathcal{H}$, where $(\ker\T_{\psi}J)^{\perp}$ is a closed subspace. Hence, all $\mu<0$ are regular values of $J$, and, for these values, we have the submanifolds
\begin{equation*}
J^{-1}(\mu)=\bigg\{\psi\in\mathcal{H}\mid -\frac{1}{2}\langle \psi\vert \psi \rangle = -\frac{1}{2}\sum\limits_{j\in\mathbb{N}}\left[q_j(\psi)^2+p^j(\psi)^2\right]=\mu\bigg\}\,,\qquad \forall \mu < 0\,.
\end{equation*}

As mentioned before, $(\Ad^*)_g = \Id_{\mathfrak{u}^*(1)}$ for every $g\in U(1)$. Therefore, the isotropy group at $\mu$ is $U(1)_{\mu}:=\{ g \in U(1) \mid  (\Ad^*)_{g}(\mu) = \mu\} = U(1)$. Let us prove that $U(1)$ acts freely and properly on $J^{-1}(\mu)$. The condition $\Phi_{e^{\ii \theta}}(\psi)=\psi$, for some $\theta\in [0,2\pi)$, implies that $\theta=0$, so that $e^{\ii\theta}=1$ is the identity element of $U(1)$ and, as a result, $\Phi$ is a free action. Moreover, $U(1)$ is homeomorphic to the group $S^1=\{z\in \mathbb{C} \mid  \vert z\vert =1 \}\subset \mathbb{C}\simeq \mathbb{R}^2$. Therefore, $U(1)$ is a compact Lie group. In view of Proposition \ref{prop:compactLiegroup}, the Lie group action $\Phi$ is also proper. Consequently, Theorem \ref{Iloc} ensures that
\begin{equation*}
P_\mu:=J^{-1}(\mu)/U(1)
\end{equation*}
is a smooth manifold for all $\mu<0$. Note that if $\pi_\mu:J^{-1}(\mu)\rightarrow P_\mu$ is such that the images of two elements of $J^{-1}(\mu)$ under $\pi_\mu$ are the same, both elements are the same up to a proportional complex constant with module one.  This amounts to saying that the images of two elements of $J^{-1}(\mu)$ are the same if their images under $\pi:\mathcal{H}_0=\mathcal{H}\setminus \{0\}\rightarrow \mathcal{P}\mathcal{H}$ are the same. As $\pi_\mu(J^{-1}(\mu))=\pi(\mathcal{H}_0)$, one can indeed identify $P_\mu=\mathcal{P}\mathcal{H}$ and see that  $\pi_{\mu}=\pi|_{J^{-1}(\mu)}$. Since the differentiable structure on $P_{\mu}$ is one making $\pi_{\mu}:J^{-1}(\mu)\rightarrow P_\mu$ into a differentiable submersion and $\pi_{\mu}=\pi|_{J^{-1}(\mu)}$, the smooth structure on $P_\mu$ is the one of $\pi(J^{-1}(\mu))=\mathcal{P}\mathcal{H}$. Hence, $P_\mu=\mathcal{PH}$  
is a smooth manifold for all $\mu<0$.   

Note that all assumptions of the Marsden--Weinstein reduction theorem are satisfied for $\mu<0$. We may now implement the reduction procedure for the symplectic form. Let us rewrite $\pi_{\mu}$ as
$$ \pi_{\mu}: J^{-1}(\mu)\ni\psi \longmapsto [\psi]\in \mathcal{PH}\,,$$
and consider the natural embedding $\jmath_{\mu}: J^{-1}(\mu)\hookrightarrow \mathcal{H}$. The symplectic form on the projective space $\mathcal{PH}$ is the only symplectic form on this latter manifold, let us say $\omega_\mu$, such that $\pi_{\mu}^*\omega_{\mu}=\jmath_{\mu}^*\omega$. In other words,
\begin{equation*}
    \omega_{\mu_{[\psi_0]}}\big({[\dot\psi]},{[\dot\phi]}\big) = \mathfrak{Im}\langle\psi\vert\phi\rangle\,, \qquad \forall\psi_0\in J^{-1}(\mu)\,,\quad \forall \dot\psi, \dot \phi \in \T_{\psi_0}(J^{-1}(\mu))\,.
\end{equation*}

\subsection{Marsden--Weinstein reduction of Schr\"odinger equations}

The vector fields $X_\alpha=X_{-\ii H_{\alpha}}$, with $\alpha=1,\ldots,r$, are tangent to $J^{-1}(\mu)$ at the points where they are defined. In fact, the self-adjointness of $H_\alpha$ implies that, on points of $D$, we have
\begin{equation*}
X_\alpha J(\psi)=-\frac{1}{2} X_{\alpha}\langle \psi \vert\psi\rangle = -\frac{\dd}{\dd t}\bigg\vert_{t=0} \langle \exp\left(-t\ii H_{\alpha}\right)\psi \vert \exp\left(-t\ii H_{\alpha}\right)\psi \rangle = 0\,.
\end{equation*}
Then, the restriction of $X_\alpha$ to $J^{-1}(\mu)$ is a well-defined vector field: it is not only tangent where defined, but it is defined on a subset that is dense in $J^{-1}(\mu)$. Let us prove this.

Since $D$ is dense in $\mathcal{H}$, for every $\psi\in J^{-1}(\mu)$ with $\mu\neq 0$, there exists a sequence $(\psi^D_n)_{n\in\mathbb{N}}$ of non-zero elements of $\mathcal{H}\backslash\{0\}$ that belong to $D$, such that $\psi^D_n \rightarrow \psi$. Note that 
$\|\psi\|^2=\langle \psi \vert \psi \rangle = -2 \mu=:k>0$. If $\psi^D_n \rightarrow \psi$, then $(\sqrt{k}\psi_n^D / \Vert \psi^D_n \Vert)_{n\in\mathbb{N}}$ is a sequence in $J^{-1}(\mu)\cap D$, such that $\sqrt{k}\psi_n^D / \Vert \psi^D_n \Vert \rightarrow \sqrt{k} \psi/\Vert\psi\Vert=\psi$. Consequently, $D_{\mu}:=J^{-1}(\mu)\cap D$ is dense in $J^{-1}(\mu)$ and $\overline{D_{\mu}}=J^{-1}(\mu)$.

The flow of $X_\alpha$, namely $F:(t,\psi)\in \mathbb{R}\times \mathcal{H}\mapsto F_t\psi:=\exp(-\ii tH_\alpha)\psi \in \mathcal{H}$, leaves invariant $J^{-1}(\mu)$. Additionally, it commutes with every $\Phi_g$, since each $\Phi_g$ is a multiplication by a phase and the $\exp(-t\ii H_\alpha)$ are linear mappings for every $t\in \mathbb{R}$. This gives rise to an induced flow on $P_\mu$ of the form $K:(t,[\psi])\in \mathbb{R}\times P_\mu\mapsto [F_t\psi]\in P_\mu$. Let us write $X_
\alpha^\mu$ for the restriction of $X_\alpha$ to $J^{-1}(\mu)$, which is tangent to this submanifold. Since the projection $\pi_{\mu}$ amounts to the restriction of $\pi$ to $J^{-1}(\mu)$, and $X^\mu_\alpha$ is the restriction of $X_\alpha$ to $J^{-1}(\mu)$, one obtains that the projection of $X^\mu_\alpha$ onto $J^{-1}(\mu)/U(1)_\mu$ is the projection of $X_\alpha$ on $\mathcal{H}_0$ onto $\mathcal{PH}$.  

To prove that the projection of $X^\mu_\alpha$ on $\mathcal{H}_0$ onto $P_\mu$ is a vector field $K_\alpha^\mu$, one has to prove that such a projection is well defined on a dense subset of $P_\mu$. In fact, Theorem \ref{Iloc} shows that $\pi_{\mu}: J^{-1}(\mu) \rightarrow J^{-1}(\mu)/U(1)_\mu\simeq \mathcal{PH}$ is a surjective submersion. Therefore,
\begin{equation*}
    \overline{\pi_{\mu}(D_{\mu})}=\pi_{\mu}(\overline{D_{\mu}})=\pi_{\mu}(J^{-1}(\mu))=J^{-1}(\mu)/U(1)\simeq \mathcal{PH}\,,
\end{equation*}
and $\pi_{\mu}(D_{\mu})$ is dense in $\mathcal{PH}$. We conclude that $K^\mu_{\alpha}$ is well-defined with a domain being a dense subset of $\mathcal{PH}$.

Let us prove that $K^\mu_\alpha$ is a Hamiltonian vector field on $P_\mu$ with a Hamiltonian function given by 
\begin{equation*}
    f_{{\alpha}}^{\mu}([\psi]) = \frac{1}{2} \langle \psi \vert H_{\alpha} \psi \rangle\,, \qquad  \qquad \forall \psi\in D(H_\alpha)\backslash\{0\}\,,\quad J(\psi)=\mu\,,
\end{equation*}
for $\alpha=1,\dotsc,r.$ 
In fact,
\begin{align*}
    \omega_{\mu}(K^\mu_\alpha,K^\mu_\beta)([\psi]) &= (\omega_{\mu})(\pi_{\mu*}X^\mu_\alpha, \pi_{\mu*}X^\mu_\beta)([\psi])=\pi^*_\mu\omega_\mu(X^\mu_\alpha,X^\mu_\beta)(\psi)\\
    &= \jmath_\mu^*\omega(X^\mu_\alpha,X^\mu_\beta)(\psi) = (\d\jmath_\mu^*f_\alpha)(X^\mu_\beta)(\psi) = (\d f^\mu_\alpha)(K^\mu_\beta)([\psi])\,.
\end{align*}
Then, $f_{{\alpha}}^{\mu}$ is the Hamiltonian function of $K^\mu_\alpha$. 

The above shows how the $t$-dependent vector field $X$ on the domain $D$ gives rise to the $t$-dependent vector field $X^{\mu}=\sum_{\alpha=1}^rb_\alpha(t)X^\mu_\alpha$ on $J^{-1}(\mu)\cap D$ projecting onto a $t$-dependent vector field $K^\mu=\sum_{\alpha=1}^rb_\alpha(t)K^\mu_\alpha$ on $\pi_\mu(J^{-1}(\mu)\cap D)\subset P^\mu$, whose integral curves are determined by the $t$-dependent projective Schr\"{o}dinger equation. 

Since $K^{\mu}$ is defined on a manifold $D_{\mu}$ modelled on a normed space, we cannot ensure straightforwardly the existence of its integral curves. Nevertheless, their existence is inferred from the existence of the integral curves of the $t$-dependent vector field $X$.

\section{Conclusions and Outlook}

This paper presents a careful symplectic geometric approach to the Schr\"odinger equations on separable Hilbert spaces determined through $t$-dependent unbounded self-adjoint Hamiltonian operators satisfying a quite general condition. The projection of such equations onto the projective space has been described as a Marsden--Weinstein reduction. The present paper describes carefully the necessary technical geometric and functional analysis details to accomplish the description of the results given. This has been accomplished in a  mainly differential geometric way and avoiding complicated functional analysis techniques.

There is much room for additional developments. In the future, we aim to study how to extend to the infinite-dimensional context several geometric techniques to deal with Schr\"odinger equations (see \cite{CCM_07,Jov_17}). Moreover, other types of Marsden--Weinstein reductions will be contemplated. We expect to analyse the Hamilton-Jacobi theory in infinite-dimensional symplectic manifolds.

\addcontentsline{toc}{section}{Acknowledgements}
\section*{Acknowledgements}
J. de Lucas acknowledges financial support by the project number 1082 of the University of Warsaw. J. Lange was partially financed by the project MTM2015-69124-REDT entitled `Geometry, Mechanics and Control' financed by the Ministerio de Ciencia, Innovaci\'on y Universidades (Spain). J. Lange and X. Rivas were partially financed by project IDUB with number PSP: 501-D111-20-2004310. J. de Lucas acknowledges a Simons-CRM professorship. X. Rivas acknowledges funding from J. de Lucas's CRM-Simons professorship, the Spanish Ministry of Science and Innovation, grants  PID2021-125515NB-C21, and RED2022-134301-T of AEI, and Ministry of Research and Universities of the Catalan Government, project 2021 SGR 00603, to accomplish a stay at the CRM which allowed him to complete some of the main parts of the present work.

\bibliographystyle{abbrv}
{\small
\bibliography{references.bib}

\end{document}